\documentclass[pra,aps,nopacs,onecolumn,twoside,superscriptaddress]{revtex4}

\usepackage{amsmath,amsfonts,amssymb,caption,color,epsfig,graphics,graphicx,hyperref,latexsym,mathrsfs,revsymb,theorem,url,verbatim,epstopdf,mathtools,enumerate}
\usepackage{subfigure,makecell,multirow,diagbox,CJK,indentfirst,cases}
\hypersetup{colorlinks,linkcolor={blue},citecolor={blue},urlcolor={red}}

\newtheorem{definition}{Definition}
\newtheorem{proposition}[definition]{Proposition}
\newtheorem{lemma}[definition]{Lemma}

\newtheorem{theorem}[definition]{Theorem}
\newtheorem{corollary}[definition]{Corollary}
\newtheorem{conjecture}[definition]{Conjecture}

\newtheorem{remark}[definition]{Remark}
\newtheorem{example}[definition]{Example}
\newtheorem{question}[definition]{Question}
\newtheorem{memo}[definition]{Memo}

\def\squareforqed{\hbox{\rlap{$\sqcap$}$\sqcup$}}
\def\qed{\ifmmode\squareforqed\else{\unskip\nobreak\hfil
\penalty50\hskip1em\null\nobreak\hfil\squareforqed
\parfillskip=0pt\finalhyphendemerits=0\endgraf}\fi}
\def\endenv{\ifmmode\;\else{\unskip\nobreak\hfil
\penalty50\hskip1em\null\nobreak\hfil\;
\parfillskip=0pt\finalhyphendemerits=0\endgraf}\fi}
\newenvironment{proof}{\noindent \textbf{{Proof.~} }}{\qed}
\def\Dbar{\leavevmode\lower.6ex\hbox to 0pt
{\hskip-.23ex\accent"16\hss}D}
\makeatletter
\def\url@leostyle{%
  \@ifundefined{selectfont}{\def\UrlFont{\sf}}{\def\UrlFont{\small\ttfamily}}}
\makeatother
\def\bcj{\begin{conjecture}}
\def\ecj{\end{conjecture}}
\def\bcr{\begin{corollary}}
\def\ecr{\end{corollary}}
\def\bd{\begin{definition}}
\def\ed{\end{definition}}
\def\bea{\begin{eqnarray}}
\def\eea{\end{eqnarray}}
\def\beq{\begin{equation}}
\def\eeq{\end{equation}}
\def\bal{\begin{aligned}}
\def\eal{\end{aligned}}
\def\bem{\begin{enumerate}}
\def\eem{\end{enumerate}}
\def\bex{\begin{example}}
\def\eex{\end{example}}
\def\bim{\begin{itemize}}
\def\eim{\end{itemize}}
\def\bl{\begin{lemma}}
\def\el{\end{lemma}}
\def\bma{\begin{bmatrix}}
\def\ema{\end{bmatrix}}
\def\bpf{\begin{proof}}
\def\epf{\end{proof}}
\def\bpp{\begin{proposition}}
\def\epp{\end{proposition}}
\def\bqu{\begin{question}}
\def\equ{\end{question}}
\def\br{\begin{remark}}
\def\er{\end{remark}}
\def\bt{\begin{theorem}}
\def\et{\end{theorem}}
\def\bmm{\begin{memo}}
\def\emm{\end{memo}}

\def\btb{\begin{tabular}}
\def\etb{\end{tabular}}

\newcommand{\nc}{\newcommand}

\def\a{\alpha}
\def\b{\beta}

\def\t{\theta}
\def\i{\iota}

\def\l{\lambda}

\def\r{\rho}
\def\s{\sigma}

\def\ps{\psi}
\def\og{\omega}

\nc{\bbA}{\mathbb{A}} \nc{\bbB}{\mathbb{B}} \nc{\bbC}{\mathbb{C}}
 \nc{\bbD}{\mathbb{D}} \nc{\bbE}{\mathbb{E}} \nc{\bbF}{\mathbb{F}}
 \nc{\bbG}{\mathbb{G}} \nc{\bbH}{\mathbb{H}} \nc{\bbI}{\mathbb{I}}
 \nc{\bbJ}{\mathbb{J}} \nc{\bbK}{\mathbb{K}} \nc{\bbL}{\mathbb{L}}
 \nc{\bbM}{\mathbb{M}} \nc{\bbN}{\mathbb{N}} \nc{\bbO}{\mathbb{O}}
 \nc{\bbP}{\mathbb{P}} \nc{\bbQ}{\mathbb{Q}} \nc{\bbR}{\mathbb{R}}
 \nc{\bbS}{\mathbb{S}} \nc{\bbT}{\mathbb{T}} \nc{\bbU}{\mathbb{U}}
 \nc{\bbV}{\mathbb{V}} \nc{\bbW}{\mathbb{W}} \nc{\bbX}{\mathbb{X}}
 \nc{\bbZ}{\mathbb{Z}}

 \nc{\bA}{{\bf A}} \nc{\bB}{{\bf B}} \nc{\bC}{{\bf C}}
 \nc{\bD}{{\bf D}} \nc{\bE}{{\bf E}} \nc{\bF}{{\bf F}}
 \nc{\bG}{{\bf G}} \nc{\bH}{{\bf H}} \nc{\bI}{{\bf I}}
 \nc{\bJ}{{\bf J}} \nc{\bK}{{\bf K}} \nc{\bL}{{\bf L}}
 \nc{\bM}{{\bf M}} \nc{\bN}{{\bf N}} \nc{\bO}{{\bf O}}
 \nc{\bP}{{\bf P}} \nc{\bQ}{{\bf Q}} \nc{\bR}{{\bf R}}
 \nc{\bS}{{\bf S}} \nc{\bT}{{\bf T}} \nc{\bU}{{\bf U}}
 \nc{\bV}{{\bf V}} \nc{\bW}{{\bf W}} \nc{\bX}{{\bf X}}
 \nc{\bZ}{{\bf Z}}

\nc{\cA}{{\cal A}} \nc{\cB}{{\cal B}} \nc{\cC}{{\cal C}}
\nc{\cD}{{\cal D}} \nc{\cE}{{\cal E}} \nc{\cF}{{\cal F}}
\nc{\cG}{{\cal G}} \nc{\cH}{{\cal H}} \nc{\cI}{{\cal I}}
\nc{\cJ}{{\cal J}} \nc{\cK}{{\cal K}} \nc{\cL}{{\cal L}}
\nc{\cM}{{\cal M}} \nc{\cN}{{\cal N}} \nc{\cO}{{\cal O}}
\nc{\cP}{{\cal P}} \nc{\cQ}{{\cal Q}} \nc{\cR}{{\cal R}}
\nc{\cS}{{\cal S}} \nc{\cT}{{\cal T}} \nc{\cU}{{\cal U}}
\nc{\cV}{{\cal V}} \nc{\cW}{{\cal W}} \nc{\cX}{{\cal X}}
\nc{\cY}{{\cal Y}}\nc{\cZ}{{\cal Z}}

\nc{\hA}{{\hat{A}}} \nc{\hB}{{\hat{B}}} \nc{\hC}{{\hat{C}}}
\nc{\hD}{{\hat{D}}} \nc{\hE}{{\hat{E}}} \nc{\hF}{{\hat{F}}}
\nc{\hG}{{\hat{G}}} \nc{\hH}{{\hat{H}}} \nc{\hI}{{\hat{I}}}
\nc{\hJ}{{\hat{J}}} \nc{\hK}{{\hat{K}}} \nc{\hL}{{\hat{L}}}
\nc{\hM}{{\hat{M}}} \nc{\hN}{{\hat{N}}} \nc{\hO}{{\hat{O}}}
\nc{\hP}{{\hat{P}}} \nc{\hR}{{\hat{R}}} \nc{\hS}{{\hat{S}}}
\nc{\hT}{{\hat{T}}} \nc{\hU}{{\hat{U}}} \nc{\hV}{{\hat{V}}}
\nc{\hW}{{\hat{W}}} \nc{\hX}{{\hat{X}}} \nc{\hZ}{{\hat{Z}}}

\nc{\hn}{{\hat{n}}}

\def\diag{\mathop{\rm diag}}

\def\max{\mathop{\rm max}}
\def\min{\mathop{\rm min}}

\def\tr{\mathop{\rm Tr}}

\def\dg{\dagger}

\newcommand{\bra}[1]{\langle#1|}
\newcommand{\ket}[1]{|#1\rangle}
\newcommand{\proj}[1]{| #1\rangle\!\langle #1 |}
\newcommand{\ketbra}[2]{|#1\rangle\!\langle#2|}

\newcommand{\norm}[1]{\lVert#1\rVert}
\newcommand{\abs}[1]{|#1|}

\def\Dbar{\leavevmode\lower.6ex\hbox to 0pt
{\hskip-.23ex\accent"16\hss}D}

\begin{document}

\normalsize

\title{Quantum Wasserstein distance between unitary operations}

\date{\today}
\author{Xinyu Qiu}\email[]{xinyuqiu@buaa.edu.cn}
\affiliation{LMIB(Beihang University), Ministry of education, and School of Mathematical Sciences, Beihang University, Beijing 100191, China}
\author{Lin Chen}\email[]{linchen@buaa.edu.cn (corresponding author)}
\affiliation{LMIB(Beihang University), Ministry of education, and School of Mathematical Sciences, Beihang University, Beijing 100191, China}
\affiliation{International Research Institute for Multidisciplinary Science, Beihang University, Beijing 100191, China}

\begin{abstract}
Quantifying the effect of noise on unitary operations is an essential task in quantum information processing. We propose the quantum Wasserstein distance between unitary operations, which shows an explanation for quantum circuit complexity and characterizes local distinguishability of multi-qudit operations. We show analytical calculation  of the distance between identity and widely-used quantum gates including SWAP, CNOT, and other controlled gates. As an application, we estimate the closeness between quantum gates in circuit, and show that the noisy operation simulates the ideal one well  when they become close under the distance. Further we introduce the $W_1$ error rate by the distance, and establish the relation between the $W_1$ error rate and two practical cost measures of recovery operation in quantum error-correction under typical noise scenarios.  
\end{abstract}

\maketitle

\section{Introduction}

Recent progress in quantum information processing derives prominent applications, such as  simulation \cite{Cattaneo2023simulation, schlimgen2021quantum},  control into computation \cite{dong2021experimental,dolde2014high} and  machine learning \cite{beer2020training, Mitarai2018circuit, Lubasch2020variational}.  Real-world imperfections exist and current quantum computers are inevitable noisy. So it is essential to characterize how much the noise  influence the implementation of quantum operations. The key point is to evaluate the similarity measure between ideal and real operations performed under noisy environments.
Generally, the similarity measure between operations can be induced by that between quantum states. The most prominent measures between  states are the trace distance induced by Schatten 1-norm \cite{Dajka2011distance}, the quantum fidelity  and the quantum relative entropy \cite{Lashkari2014relative}. They are all unitarily invariant, and it is not always desirable for certain applications like quantum error-correction. Recently, the quantum Wasserstein distance between quantum states has been proposed, which recovers the classical Wasserstein distance for quantum states diagonal in the canonical basis   \cite{de2021the}. It derives numerous applications, such as quantum differential privacy \cite{2203.03591}, quantum concentration inequality \cite{DePalma2022concentration}, quantum circuit complexity \cite{2208.06306}. One problem arises when comparing the operations by the distance between a single couple of input and output states. That is, the distance is always state-dependent and  characterizing the operation requires numerous input states \cite{chen2016entangling}. Hence, it is necessary to construct the similarity measure for operations. This is the motivation of this paper.

The  similarity measure between operations can be constructed using the idea of the discrimination of unitary transformations \cite{A2001Statistical}, which is an important application of quantum state discrimination \cite{Joonwoo2015quantum}. 
The  distinguishability between operations can be quantified by a certain distance between two output states. The maximization or average is taken over all input states to make it state-independent. Several distances have been employed to construct the measure, such as  the trace distance \cite{Ariano2001quantum,A2001Statistical}, Schatten 2-norm \cite{Chen2022quantum}, diamond norm \cite{Regula2021operational} and other measures \cite{Gilchrist2005diatance}. The trace distance shows a compelling physical interpretation for the probability in positive operator valued measurement (POVM). The Schatten 2-norm can be efficiently estimated in quantum circuits.  The diamond norm is a widely-used figure of merit to evaluate the threshold for fault-tolerant quantum computation. 
These distances show global distinguishability between quantum states. Neither of them allows to distinguish the states that differ locally nor relate to the circuit complexity. So we focus on the quantum Wasserstein distance between operations, and it will show the above properties.

In this paper, we introduce a similarity measure of unitary operations, named the quantum Wasserstein distance between unitary operations. It is a state-independent measure of the distance between operations and a diagnostic of noise. Induced by the quantum Wasserstein distance between states, it measures distinguishability regarding to extensive and quasilocal observables, and shows an explanation for quantum circuit complexity.  We show the basic properties of the distance, such as faithfulness, symmetry, and right unitary invariance. We investigate the calculation of the distance, and show some analytical results for the distance between the identity and some widely-used unitary operations including CNOT, SWAP, and generalized controlled gate. We show two applications of the distance. First we consider the distance in quantum circuits and apply it to estimate the closeness between two sequences of gates, and show that the noisy operation simulates the ideal operation well when they become close under the distance. Next we introduce the $W_1$ gate error rate by the distance, which quantifies the realization of quantum gates under noisy environment. We establish the relation between the $W_1$ error rate and two real cost measures of recover operation, including circuit cost and experiment cost. Hence the $W_1$ error rate is related to the practical cost of eliminating the effect of noise on a specific type of gate. The lower the  $W_1$ error rate of a noisy gate is, the less it may cost to implement its recovery operation.

Building quantum computers derives a strong need for accurate characterization of the noise in quantum gate implementations. Gate error rate \cite{sanders2015bounding} and  fidelity \cite{nielsen2002fidelity,lu2020direct} are the most widely used figure of merit for the performance of a single quantum gate. Gate fidelity is experimentally convenient, while the connection of that with fault-tolerance requirements is not direct. Gate error rate that induced by diamond norm is an alternative bound that can yield tighter estimates of gate performance. Although no measure can be generally suitable for all quantum information processing tasks, they contribute to understanding and improving specific aspects of the quantum operations. The $W_1$ error rate we proposed characterizes the noisy implementation of quantum gates from the perspective of cost measures for their recovery operations in quantum error-correction.

The rest of the paper is organized as follows. In Sec. \ref{Sec:preliminaries}, we show the notations, some properties of the quantum $W_1$ distance between states, and the formalism of average gate fidelity and error rate. In Sec. \ref{Sec:def, pro, cal}, we show the definition, properties and calculation of quantum $W_1$ distance between operations. In Sec. \ref{Sec:app to qc}, we show an application, i.e. the estimation of the closeness between operations in quantum circuits. In Sec.  \ref{sec:gate error rate}, we introduce the $W_1$ gate error rate with the help of   quantum $W_1$ distance between operations, and show the noisy implementation of arbitrary single-qubit gate and CNOT gate under typical noise scenarios. We conclude in Sec. \ref{Sec:conclusion}.

\section{Preliminaries and notations}
\label{Sec:preliminaries}
In this section, we show the notations and some facts used in this paper. In Sec. \ref{sec:notation}, we present the notations of this paper.  In Sec. \ref{sec:W1 distance between states}, we show the definition and some properties of quantum $W_1$ distance between states. In Sec. \ref{sec:preliminary error rate}, we introduce the derivation of the average gate fidelity and gate error rate induced by different kinds of norms.
\subsection{Notations}
\label{sec:notation}
  We denote the set of traceless,  self-adjoint linear operators by $\cM_n$,  the set of $n$-qudit quantum states  by $\cS_n$, the set of unitary operations acting on $n$-qubit states as $\cU_n$, and by $\cP_n$ the set of the probability distributions on $[d]^n$.

Some well-known single-qubit gate include the Hadamard gate $H=\frac{1}{\sqrt{2}}\bma 1&1\\
1&-1\ema$, and the Pauli matrices
$\s_x=\bma0&1\\1&0\ema$, $\s_y=\bma0&-i\\i&0\ema$, 
$\s_z=\bma1&0\\0&-1\ema$.
 The two-qubit gates include  
the CNOT gate 
$U_{CN}=\bma1&0&0&0\\
0&1&0&0\\
0&0&0&1\\
0&0&1&0\ema$, the controlled-Z gate 
$U_{CZ}=\bma1&0&0&0\\
0&1&0&0\\
0&0&1&0\\
0&0&0&-1\ema$, the SWAP gate
$U_{SW}=\bma
1&0&0&0\\
0&0&1&0\\
0&1&0&0\\
0&0&0&1
\ema$, and the generalized controlled phase gate $U_{CP}=\bma
1&0&0&0\\
0&1&0&0\\
0&0&1&0\\
0&0&0&e^{i\t}
\ema$.
The single-qudit Pauli gate $X=\sum_{q=0}^{d-1}\ket{q\oplus 1}\bra{q}$ and $Z=\sum_{q=0}^{d-1}\og^q\ket{q}\bra{q}$, for $\og=e^{2\pi i/d}$.

The Schatten $p$-norm for arbitrary matrix $A\in \bbC^{N\times M}$ and $p\in [1,\infty)$ is defined as $\norm{A}_p=\big[\tr(A^\dg A)^{\frac{p}{2}}\big]^\frac{1}{p}$. By setting $p=1$ into the definition of $\norm{\cdot}_p$, one has the Schatten 1-norm (trace norm) given by $\norm{A}_1=\tr\sqrt{A^\dg A}$, which is equal to the sum of singular values of $A$. For two states $\r,\s\in\cS_n$, $\norm{\r-\s}_1$ is typically denoted as the trace distance between $\r$ and $\s$.

\subsection{The quantum Wasserstein distance of order 1 between states}
\label{sec:W1 distance between states}
The well-known similarity measures between quantum states including the trace distance, quantum fidelity and relative entropy are all unitarily invariant. They characterize the global distinguishability of states. For certain applications, such as quantum error correction and quantum machine learning, it is desirable to use the distance with respect to which the state $\ket{0}^{\otimes n}$ is much closer to $\ket{1}\otimes \ket{0}^{\otimes (n-1)}$ than $\ket{1}^{\otimes n}$. Such a distance is called  the quantum  Wasserstein distance of order 1 \cite{de2021the}. For convenience, we denote quantum Wasserstein distance of order 1 as the quantum $W_1$ distance in the context. It can recover the hamming distance for vectors of the canonical basis, and more generally robustness against local perturbations on the input states.

We show the definition and some important properties of the quantum Wasserstein distance, which will be used in this paper.
First, we show some basic definitions. 
 The quantum Wasserstein norm of order 1 is a kind of unique norm on $\cM_n$. It is defined as follows \cite{de2021the},
 \begin{definition}
 	We define the quantum $W_1$ norm on $\cM_n$ as, for any $X\in\cM_n$,
 \begin{eqnarray}
 \norm{X}_{W_1}=\frac{1}{2}\min \bigg(\sum_{i=1}^n\norm{X^{(i)}}_1:X^{(i)}\in\cM_n, \tr_iX^{(i)}=0, X=\sum_{i=1}^nX^{(i)} 
 \bigg),
 \end{eqnarray}
	where $\tr_i [\cdot]$ denotes the partial trace over the i-th subsystem.
 \end{definition}
Following the quantum $W_1$ norm, the quantum Wasserstein distance  between states is naturally obtained \cite{de2021the}.
\begin{definition}
	The quantum Wasserstein distance of order 1 between two quantum states $\r$, $\s$ is defined as,
	\begin{eqnarray}
		\label{def:w1diatance}
		W_1(\r,\s)=&&\norm{\r-\s}_{W_1}\nonumber\\
		=&&\min \left\{\sum_{i=1}^n c_i: c_i\geq0, \r-\s=\sum_{i=1}^nc_i(\r^{(i)}-\s^{(i)}), \r^{(i)},\s^{(i)}\in\cS_n, \tr_i \r^{(i)}=\tr_i\s^{(i)} \right\}.
	\end{eqnarray} 
\end{definition}
Next we list the properties of the available quantum $W_1$ distance. They will be used for the derivation and applications of the quantum $W_1$ distance between operations.

The following fact shows that the quantum $W_1$ norm keeps the same upper and lower bounds in terms of the trace norm as its classical counterpart. It establishes the relation between the quantum $W_1$ norm and other Schatten $p$-norms.
\begin{lemma}
	\label{le:relation with trace norm}
	(relation with the trace norm, \cite{de2021the}) For any $X\in\cM_n$, 
	\begin{eqnarray}
		\frac{1}{2}\norm{X}_1\leq\norm{X}_{W_1}\leq\frac{n}{2}\norm{X}_{1}.
	\end{eqnarray}	
	Moreover, if $\tr_iX=0$ for some $i\in[n]$, then 
	\begin{eqnarray}
		\label{eq:w1normX}	\norm{X}_{W_1}=\frac{1}{2}\norm{X}_{1},
	\end{eqnarray}
	i.e., for any $\r,\s\in\cS_n$ such that $\tr_i\r=\tr_i\s$ for some $i\in[n]$,
	\begin{eqnarray}
		\label{eq:w1normr-s}
		\norm{\r-\s}_{W_1}=\frac{1}{2}\norm{\r-\s}_1.
	\end{eqnarray}
\end{lemma}

Lemma \ref{le:tensorization} and Corollary \ref{cor:geqnorm1} show that the quantum $W_1$ distance is additive with respect to the tensor product and its counterpart. This property can not be satisfied by the trace distance. In this paper, they are used for calculating the quantum $W_1$ distance between operations and deriving its properties. 
\begin{lemma}
	\label{le:tensorization}
(tensorization, \cite{de2021the}) For any  $X\in\cM_n$, 
\begin{eqnarray}
	\norm{X}_{W_1}\geq\norm{\tr_{m+1...m+n}X}_{W_1}+\norm{\tr_{1...m}X}_{W_1},
\end{eqnarray}
and for any $n$-qudit states $\r,\s\in\cS_{m+n}$,
\begin{eqnarray}
	\norm{\r-\s}_{W_1}\geq
	\norm{\r_{1...m}-\s_{1...m}}_{W_1}+	\norm{\r_{m+1...m+n}-\s_{m+1...m+n}}_{W_1}.
\end{eqnarray}
Moreover, for any $\r',\s'\in\cS_m$ and $\r'',\s''\in\cS_n$, 
\begin{eqnarray}
	\label{eq:r'Or''}
	\norm{\r'\otimes \r''-\s'\otimes\s''}_{W_1}
	=\norm{\r'-\s'}_{W_1}+\norm{\r''-\s''}_{W_1}.\end{eqnarray}
\end{lemma}
 
\begin{corollary}
	\label{cor:geqnorm1}
	(lower bound for $W_1$ distance, \cite{de2021the})
 For any $\r,\s\in\cS_n$,
\begin{eqnarray}
	\label{ieq:tensorization}
	\norm{\r-\s}_{W_1}\geq\frac{1}{2}\sum_{i=1}^n\norm{\r_i-\s_i}_1,
\end{eqnarray}		and equality holds whenever both $\r$ and $\s$ are product states.
\end{corollary} 
 
The following  observation states that the quantum $W_1$ distance recovers the classical $W_1$ distance for the quantum states diagonal in the canonical basis. It contributes to the calculation of the quantum $W_1$ distance between operations.
\begin{lemma}
	\label{le: recovery of classical W1}
(recovery of the classical $W_1$ distance, \cite{de2021the}) Let $p,q\in\cP_n$, and let
\begin{eqnarray}
	\r=\sum_{x\in[d]^n}p(x)\ketbra{x}{x},
	\quad
	\s=\sum_{y\in[d]^n}q(y)\ketbra{y}{y}.
\end{eqnarray}
Then,
\begin{eqnarray}
	\label{eq:recoveryofW1}
	\norm{\r-\s}_{W_1}=W_1(p,q).
\end{eqnarray}
In particular, the quantum $W_1$ distance between vectors of the canonical basis coincides with the Hamming distance:
\begin{eqnarray}
	\label{eq:projx-projy}
	\norm{\ketbra{x}{x}-\ketbra{y}{y}}_{W_1}=h(x,y),
	\quad
	x,y\in[d]^n.
\end{eqnarray}
Here the Hamming distance between $x,y\in[d]^n$ is the number of different components:
\begin{eqnarray}
	h(x,y)=\abs{\{i\in[n]:x_i\neq y_i\}},
\end{eqnarray}
\end{lemma}
where $x=(x_1,...,x_n)^T$ and  $y=(y_1,...,y_n)^T$.

\subsection{Average gate fidelity and gate error rate}
\label{sec:preliminary error rate}
In practice, quantum gates can be hardly isolated from the environment and the gate-noise interaction can transform the ideal gate  into actual gate.  The implementation of actual gate may lead to information leakage from the quantum system. So it is important to estimate how much the actual gates can affect the states in the system. Average gate fidelity \cite{nielsen2002fidelity} is firstly proposed to accomplish such a task. 
Suppose the ideal quantum gate $U$ acts on the input state $\r\in \cS_n$ and performs the action $\r\rightarrow\r_{id}=U\r U^\dg$. The actual gate implemented on the input state is denoted by the channel $\cV$, which acts as $\r\rightarrow \r_{ac}=\cV(\r)$. Averaging over pure state input with respect to the Haar measure derives the average gate fidelity,
\begin{eqnarray}
\eta:=\int d\mu(\r)\tr[(U\r U^\dg)\cV(\r)].
\end{eqnarray} 
By now, the relation between average gate fidelity and fault tolerance requirements is insufficient. To overcome this problem,  gate error rate \cite{Fuchs1999Cryptographic,magesan2011scalable,sanders2015bounding} is proposed, whose upper bound is an appropriate measure to assess progress towards fault tolerant quantum computation. 
The gate error rate can be derived by the error rate of probability distributions $d(p_{id}, p_{ac})=\frac{1}{2}\sum_{x\in X}\abs{p_{id}(x)-p_{ac}(x)}$, where $p_{id}(x)$ ($p_{ac}(x)$) corresponds to the probability of a ideal (actual) output over the set of all possible outcomes. 
 We compare these states with the help of POVM $\{M_m\}$. The error rate of this measurement is $d(P_U^{(m)}, P_V^{(m)})$, where  $P_U^{(m)}=\tr(M_m\r_{id})$ and $P_V^{(m)}=\tr(M_m\r_{ac})$. Taking the maximization of $d(P_U^{(m)}, P_V^{(m)})$ over all possible choices of measurement, the following probability error rate induced by Schatten $1$-norm  is obtained,
$d_1(\r_{ac},\r_{id}):=\frac{1}{2}\norm{\r_{ac}-\r_{id}}_1$.
The error rate $U$ can be defined by taking the maximization over all input states as \cite{Fuchs1999Cryptographic} 
\begin{eqnarray}
	\label{def:e1(u,v)}
	e_1(U,\cV):=\max_\r d_1(\r_{ac},\r_{id})=\frac{1}{2}\max_{\r}\norm{U\r U^\dg-\cV(\r)}_1.	
\end{eqnarray}
Amending the above definition by maximizing over inputs $\r$ and ancillary spaces using the diamond norm, another definition of gate error rate is derived \cite{sanders2015bounding},
\begin{eqnarray}
	\label{def:ed(u,v)}
	e_\diamond(U,\cV):=\frac{1}{2}\norm{U\r U^\dg-\cV(\r)}_\diamond.
\end{eqnarray} 
Note that the average gate fidelity and gate error rate is not directly connected. The reported fidelity alone implies loose bounds on the gate error rate. The tighter bounds or more direct relation with fault tolerance computation is possible when choosing other performance measures. Following this idea, we will propose the $W_1$ error rate by the $W_1$ distance between quantum states, and establish its lower bounds with the help of $W_1$ distance between operations. It will be presented in Sec. \ref{sec:gate error rate}.
\section{The definition, properties and calculation of $\cD(U,V)$}
\label{Sec:def, pro, cal}
In this section, we propose the quantum $W_1$ distance between unitary operations $U,V$ by the quantum $W_1$ norm, where $U$ and $V$ are two unitary operations acting on the same state space. In Sec. \ref{sec:properties}, we show some basic properties of $\cD(U,V)$. In Sec. \ref{sec:calculation}, we show some analytical calculations of the distance including the distance between the identity and some widely-used unitary operations.

We show the definition of $\cD(U,V)$. It is given by taking the maximization over all states in terms of the quantum Wasserstein distance in Definition \ref{def:w1diatance}.
\begin{definition}
	\label{def:D(U,V)}
Given two unitary operations $U,V\in\cU_n$ acting on the $n$-qudit state, their quantum Wasserstein distance $\cD(U,V): \cU_n\times \cU_n\rightarrow \bbR$ is the maximal quantum $W_1$ distance between the states  they have performed on,
\begin{eqnarray}
	\cD(U,V)=\max_{\r\in\cS_n}\norm{U\r U^\dg-V\r V^\dg}_{W_1}.
\end{eqnarray}
By the convexity of the quantum $W_1$ norm, we need only take the maximization over all pure states,
\begin{eqnarray}
	\label{def:E(U,V)}
	\cD(U,V)=\max_{\proj{\psi}\in\cS_n}\norm{U\ketbra{\psi}{\psi}U^\dg-V\ketbra{\psi}{\psi}V^\dg}_{W_1},
\end{eqnarray}
where the maximum is over all normalized states $\ket{\psi}$ in the state space $\cS_n$.  
\end{definition}

The quantum $W_1$  distance above shows the explanation for quantum circuit complexity \cite{2208.06306}. That is, the distance $\cD(U,V)$ shows the lower bound for the minimum number of gates (smallest circuit) that is required to transform operations $U$ and $V$ to each other. This property of $\cD(U,V)$  will be utilized in Sec. \ref{sec:gate error rate}.

On the other hand, it has been shown that the quantum $W_1$ distance between states allows to distinguish quantum states that differ locally in \cite{de2021the}. We show that quantum $W_1$ distance between operations characterizes local distinguishability of operations in multi-qubit scenario, as it is induced by $\norm{\cdot}_{W_1}$. Other distance induced by   
the measure that is unitarily invariant can not show such a property \cite{de2021the}. So the quantum $W_1$ distance between operations is a unique distance showing the local difference of operations.
 We illustrate the above viewpoint by an example for the two-qudit operation.   One can obtain that  $\cD(I^{\otimes 2},I\otimes X)=1$, see Proposition \ref{pro: D(I,sx)=k}.  Adding the Pauli gate $X$ locally on the second qudit increases the $W_1$ distance between identity and the total operation, i.e. $\cD(I^{\otimes 2},X^{\otimes 2})=2$. So the local distinguishability between the operations can be characterized. Such local difference can not be detected by other distances such as the Schatten 1-norm or fidelity. In fact, we have $\max_\r\norm{\r-(I\otimes X)\r(I\otimes X)}_1=\max_\r\norm{\r-(X\otimes X)\r(X\otimes  X)}_1=2$, and $\min_\r F(\r,(I\otimes X)\r(I\otimes X))=\min_\r F(\r,(X\otimes X)\r(X\otimes X))=0$. 
 The quantum $W_1$ distance  between nonlocal operations can also characterize their local property. For example, we consider the distance for two nonlocal qubit gates CNOT and $Q=(I\otimes \s_x)U_{CN}(\s_x\otimes I)$ that are locally different. We obtain that their difference can be characterized by the $W_1$ distance, i.e.   $\cD(I,\rm CNOT)=\sqrt{2}$ and $\cD(I,Q)=2$, see Propositions \ref{pro:D(I,CNOT)} and \ref{pro:permutation matrix}. Their local distinguishability can not described by other distance, as we have $\max_\r\norm{\r-U_{CN}\r U_{CN}}_1=\max_\r\norm{\r-Q\r Q^\dg}_1=2$  and $\min_\r F(\r,U_{CN}\r U_{CN})=\min_\r F(\r,Q\r Q^\dg)=0$.

\subsection{Some properties of $\cD(U,V)$}
\label{sec:properties}
For the convenience of deriving the applications of the quantum $W_1$ distance between operations, we present some basic properties of it and show the proof as follows. 
\begin{proposition}
	\label{pro:property}
The quantum Wasserstein distance $\cD(U,V)$ between unitary operations $U$ and $V$ satisfies the following properties:
\begin{enumerate}
	\item 
	\label{eq:faithfulness}
	Faithfulness: $\cD(U,V)=0$ if and only if $U=V$;
	
	\item 
	\label{eq:symmetry}
	Symmetry:  $\cD(U,V)=\cD(V,U)$;
	
	\item 
	\label{eq:triangle}
	Triangle inequality: $\cD(U,V)\leq \cD(U,M)+\cD(M,V)$, for $U,V,M\in\cU_n$;
	
	\item 
		\label{eq:rightinvariant}
	Right unitary invariance: $\cD(UM,VM)=\cD(U,V)$, for $U,V,M\in\cU_n$;
	
	\item
	\label{eq:leftinvariant}
	$\cD(NU,NV)=\cD(U,V)$, for $U,V\in\cU_n$ and $N\in \cU_1^{\otimes n}$;
	 	\item 
	 \label{eq:bounds}
	 Bounds: $0\leq \cD(U,V)\leq n$, for $U,V\in\cU_n$;

	 \item 
	 \label{eq:E(I,Udg)}
	Conjugate transpose invariance with identity: $\cD(I,U)=\cD(I,U^\dg)$;
	 
	\item 
	\label{eq:E(U_2U_1,V_2V_1)leq}
	$\cD(U_2U_1,V_2V_1)\leq \cD(U_1^\dg,U_2)+\cD(V_1^\dg,V_2)$;
	
	\item  
	\label{eq:superadditivity}
	Superadditivity under tensorization: $\cD(U_1\otimes U_2,V_1\otimes V_2)\geq \cD(U_1,V_1)+\cD( U_2, V_2)$;
	
	\item 
	\label{eq:E(U1OU2,V1OV2)}
	$\cD(U_1\otimes U_2,V_1\otimes V_2)\leq \cD(U_1\otimes I,V_1\otimes I)+\cD(I\otimes U_2, I\otimes V_2)$.
\end{enumerate}
\end{proposition}
\begin{proof}
The first two properties follow from the  faithfulness and symmetry of the quantum $W_1$ norm, respectively.

Property \ref{eq:triangle} follows from the triangle inequality of the quantum $W_1$ norm,
\begin{eqnarray}
\norm{U\r U^\dg-V\r V^\dg}_{W_1}\leq
\norm{V\r V^\dg-M\r M^\dg}_{W_1}
+\norm{M\r M^\dg-U\r U^\dg}_{W_1}.
\end{eqnarray}
The equality holds when $M=e^{i\t}U$, $M=e^{i\varphi}V$ or $V\r V^\dg-M\r M^\dg=k(M\r M^\dg-U\r U^\dg)$, for $k\geq 0$.

Property \ref{eq:rightinvariant} can be proved as follows,
\begin{eqnarray}
	&&\cD(UM,VM)\\
	=&&\max_{\ket{\ps}}\norm{UM\proj{\ps} M^\dg U^\dg-VM\proj{\ps} M^\dg V^\dg}_{W_1}\\
	=&&\max_{\ket{\xi}}\norm{U\proj{\xi} U^\dg-V\proj{\xi} V^\dg}_{W_1},
\end{eqnarray}
where $\ket{\xi}=M\ket{\ps}$ is any pure state. Thus the maximization takes over all pure state. The last equality is equal to $\cD(U,V)$.

Property \ref{eq:leftinvariant} can be obtained as the quantum $W_1$ distance is invariant with respect to unitary operations acting on a single qudit.

Property \ref{eq:bounds} is obtained with the help of property \ref{eq:rightinvariant} by choosing $M=U^\dg$,
\begin{eqnarray}
	\cD(U,V)
	\leq&&\max_{\ket{\ps}}\left\{\frac{n}{2}\norm{U\ketbra{\psi}{\psi}U^\dg-V\ketbra{\psi}{\psi}V^\dg}_1\right\}\\
	=&&\max_{\ket{\ps}}\left\{n\sqrt{1-\abs{\bra{\psi}U^\dg V\ket{\psi}}^2} \right\}
	= n,
\end{eqnarray}
where the inequality comes from the fact in Lemma \ref{le:relation with trace norm}. On the other hand, $\cD(U,V)\geq 0$ is obtained directly from the nonegativity of the quantum $W_1$ norm.  So the desired result is obtained.

Property \ref{eq:E(I,Udg)} is proved as follows,
\begin{eqnarray}
\cD(I,U)=&&\max_{\ket{\psi}}\norm{\proj{\psi}-U\proj{\psi}U^\dg}_{W_1}\\
=&&\max_{\ket{\xi}=U\ket{\psi}}\norm{\proj{\xi}-U^\dg\proj{\xi}U}_{W_1}\\
=&&\cD(I,U^\dg).
\end{eqnarray}

Property \ref{eq:E(U_2U_1,V_2V_1)leq} is proved with the help of properties \ref{eq:triangle} and \ref{eq:rightinvariant}.  One can obtain that
\begin{eqnarray}
	&&\cD(U_2U_1,V_2V_1)\\
=&&
\max_{\r}\norm{U_2U_1\r U_1^\dg U_2^\dg-V_2V_1\r V_1^\dg V_2^\dg}_{W_1}\\
\leq&&
\max_{\r}\norm{U_2U_1\r U_1^\dg U_2^\dg-\r}_{W_1}+\norm{\r-V_2V_1\r V_1^\dg V_2^\dg}_{W_1}\\
=&&\cD(I,U_2U_1)+\cD(I,V_2V_1)\\
=&&\cD(U_1^\dg,U_2)+\cD(V_1^\dg,V_2).
\end{eqnarray}
The equality holds when $U_2U_1=e^{i\t}I$,  $V_2V_1=e^{i\varphi}I$, or $U_2U_1\r U_1^\dg U_2^\dg-\r=k(\r-V_2V_1\r V_1^\dg V_2^\dg)$, for $k\geq 0$.

Property \ref{eq:superadditivity} holds from the tensorization of the quantum $W_1$ norm in Lemma \ref{le:tensorization}.
 We have
 \begin{eqnarray}
 &&\norm{(U_1\otimes U_2)\r (U_1^\dg\otimes U_2^\dg)-(V_1\otimes V_2)\r (V_1^\dg\otimes V_2^\dg)}_{W_1}\\
 &&\geq\norm{U_1\r_1U_1^\dg-V_1\r_1V_1^\dg}_{W_1}+\norm{U_2\r_2U_2^\dg-V_2\r_2V_2^\dg}_{W_1},
 \end{eqnarray}
where $\r_i$ denotes the corresponding reduced state.
Take the maximum of all pure states $\r,\r_1,\r_2$ on both sides of the above inequality. Then																														the property can be proved.

Next we prove property \ref{eq:E(U1OU2,V1OV2)} by the triangle inequality of the $W_1$ norm. Let $\s=(U_1\otimes U_2)\r (U_1^\dg\otimes U_2^\dg)$. We have
\begin{eqnarray}
&&\norm{(U_1\otimes U_2)\r (U_1^\dg\otimes U_2^\dg)-(V_1\otimes V_2)\r( V_1^\dg\otimes V_2^\dg)}_{W_1}\\
=&&\norm{\s-(V_1\otimes V_2)(U_1^\dg\otimes U_2^\dg)\s (U_1\otimes U_2)(V_1^\dg\otimes V_2^\dg)}_{W_1}\\
\leq&&\norm{\s-(V_1U_1^\dg\otimes I)\s (U_1 V_1^\dg\otimes I)}_{W_1}\\
+&&\norm{(V_1U_1^\dg\otimes I)\s (U_1 V_1^\dg\otimes I)-
(I\otimes V_2U_2^\dg)(V_1U_1^\dg\otimes I)\s (U_1 V_1^\dg\otimes I)(I\otimes U_2V_2^\dg)}_{W_1}\\
=&&\norm{(U_1\otimes I)\eta(U_1^\dg\otimes I)-(V_1\otimes I)\eta (V_1^\dg\otimes I)}_{W_1}
\\
+&&\norm{(I\otimes U_2)\mu(I\otimes U_2^\dg)-(I\otimes V_2)\mu(I\otimes V_2^\dg)}_{W_1}
\\
\leq&&\cD(U_1\otimes I,V_1\otimes I)+\cD(I\otimes U_2, I\otimes V_2),
\end{eqnarray}
where $\eta=(U_1^\dg \otimes I)\s(U_1 \otimes I)$, and $\mu=(I\otimes U_2^\dg)(V_1U_1^\dg\otimes I)\s (U_1 V_1^\dg\otimes I)(I\otimes U_2)$. 
Hence it holds that $\cD(U_1\otimes U_2,V_1\otimes V_2)\leq \cD(U_1\otimes I,V_1\otimes I)+\cD(I\otimes U_2, I\otimes V_2)$.

The condition that the equality in properties \ref{eq:triangle} and \ref{eq:E(U_2U_1,V_2V_1)leq} hold comes from the following fact. For any norm induced by inner product, it can be proved that $\norm{x+y}\leq \norm{x}+\norm{y}$. The equality holds when $y=0$ or $x=ay$, $a\geq 0$. Hilbert space is the inner product space, and $\norm{\cdot}_{W_1}$ defined on the subspace of that follows the above fact.
\end{proof}

\subsection{The analytical calculation of $\cD(U,V)$}
\label{sec:calculation}
By the definition of quantum $W_1$ distance between states in (\ref{def:w1diatance}), calculating the distance analytically is a challenge \cite{de2021the}. 
The derivation of $\cD(U,V)$ in Definition \ref{def:D(U,V)} requires to take the maximization over all pure states with respect to  the quantum $W_1$ distance between states. So the analytical calculation of the $W_1$ distance between operations is more challenging than that of two states. Our calculation may provide inspiration for deriving the distance between any two unitary operations, make the quantum $W_1$ distance applicable and induce more applications.
 The results illustrate the local distinguishability of nonlocal gates, which can not be detected by fidelity or Schatten 1-norm.   We show the analytical results of the quantum $W_1$ distance between single qubit operations in Sec. \ref{sec:case n=1}, some two-qubit operations in Sec. \ref{sec:case n=2} and the multi-qubit operations in \ref{sec:case any n}.

\subsubsection{The $W_1$ distance for single-qudit operations}
\label{sec:case n=1}
We consider the quantum $W_1$ distance between arbitrary single-qubit operations, and obtain the following fact.
\begin{proposition}
The quantum $W_1$ distance between  single-qubit operations $U,V$ in $d$-dimensional Hilbert space is equal to  $\sqrt{\frac{1}{2}(1-\cos\a)}$ for $d=2$, and $\sqrt{1-\min_{\sum_j\abs{a_j}^2=1} \abs{\sum_j\abs{{a_j}}^2e^{i\a_j}}^2}$ for $d>2$. Here according to the results presented in \cite{huang2022query},
 \begin{eqnarray}
	\min_{\sum_j\abs{a_j}^2=1} \abs{\sum_j\abs{{a_j}}^2e^{i\a_j}}=
	\begin{cases}
		\cos\frac{\Theta(U^\dg V)}{2}&0\leq\Theta(U^\dg V)<\pi,
		\\
		0&\Theta(U^\dg V)\geq\pi,
	\end{cases}
\end{eqnarray}
where $\Theta(U^\dg V)$ denotes the length of the smallest arc containing all the eigenvalues of unitary operation $U^\dg V$ on the unit circle.
\end{proposition}

\begin{proof}
Using Lemma \ref{le:relation with trace norm}, we have
\begin{eqnarray}
	\cD(U,V)
	=&&\frac{1}{2}\max_{\ketbra{\psi}{\psi}\in\cS_1}\norm{U\ketbra{\psi}{\psi}U^\dg- V\ketbra{\psi}{\psi}V^\dg}_{1}\\
	=&&\max_{\ketbra{\psi}{\psi}\in\cS_1}\sqrt{1-\abs{\bra{\psi}U^\dg V\ket{\psi}}^2},
\end{eqnarray}
where the first equality holds because the quantum $W_1$ norm is invariant with respect to unitary operations acting on  single qubit. The second equality comes from Lemma \ref{le:relation with trace norm}. The third  one comes from $\norm{\a uu^*-\b vv^*}_1=\sqrt{(\a+\b)^2-4\a\b\abs{\langle u,v\rangle}^2}$.

First we consider the case for $d=2$, i.e., the operations act in the two-dimensional space.  In order to obtain $\cD(U,V)$, it suffices to compute the minimum of $\abs{\bra{\psi}U^\dg V\ket{\psi}}$. 
Let $U^\dg V=R^\dg D_2R$ be the spectral decomposition of $U^\dg V$, where $D_2=\diag\{1,e^{i\a}\}$ and $R$ is an unitary matrix. The state  $\ket{\xi}=R\ket{\psi}$ is an arbitrary one-qubit state. Suppose $\ket{\xi}=\cos\frac{\t}{2}\ket{0}+e^{i\phi}\sin\frac{\t}{2}\ket{1}$, for $0\leq\t\leq\pi, 0\leq\phi<2\pi$. We have
\begin{eqnarray}
	\label{eq:calculate2dim}
	\cD(U,V)
	=&&\sqrt{1-\min_{\t,\phi}\abs{\bra{\xi}D_2\ket{\xi}}^2}\\
	=&&\sqrt{1-\min_{\t}\abs{\cos^4\frac{\t}{2}+\sin^4\frac{\t}{2}+2\cos^2\frac{\t}{2}\sin^2\frac{\t}{2}\cos\a}}\\
	=&&\sqrt{\frac{1}{2}(1-\cos\a)}.
\end{eqnarray}

Then we consider the case for $d>2$. As we have shown above, we assume that $U^\dg V=R_d^\dg D_dR_d$ is the spectral decomposition of $U^\dg V$, where $D_d=\diag\{e^{i\a_0},e^{i\a_1},...,e^{i\a_{d-1}}\}$ and $R_d$ is a $d\times d$ unitary matrix. Suppose $\ket{\ps}$ is an arbitrary qudit state and $\ket{\eta}=R_d\ket{\ps}=\sum_ja_j\ket{j}$, for  $\sum_j\abs{a_j}^2=1$.  We have
\begin{eqnarray}
	\cD(U,V)
	=&&\sqrt{1-\min_{\ket{\eta}}\abs{\bra{\eta}D\ket{\eta}}^2}\\
	=&&\sqrt{1-\min_{\sum_j\abs{a_j}^2=1} \abs{\sum_j\abs{{a_j}}^2e^{i\a_j}}^2}.
\end{eqnarray}	
It implies that once we have obtained the eigenvalues of $U^\dg V$, the calculation of  $\cD(U,V)$ is transformed into an optimization problem. The optimization is equivalent to minimize the convex sum of the eigenvalues of $U^\dg V$. 
According to the results presented in \cite{huang2022query}, one has
\begin{eqnarray}
	\min_{\sum_j\abs{a_j}^2=1} \abs{\sum_j\abs{{a_j}}^2e^{i\a_j}}=
	\begin{cases}
		\cos\frac{\Theta(U^\dg V)}{2}&0\leq\Theta(U^\dg V)<\pi,
		\\
		0&\Theta(U^\dg V)\geq\pi,
	\end{cases}
\end{eqnarray}
where $\Theta(U^\dg V)$ denotes the length of the smallest arc containing all the eigenvalues of unitary operation $U^\dg V$ on the unit circle.
\end{proof}

\subsubsection{The $W_1$ distance for two-qubit operations}
\label{sec:case n=2}
We consider the two-qubit unitary operations $U$ and $V$ in two-dimensional space.   Using Property \ref{eq:rightinvariant} of the quantum $W_1$ distance between operations, calculating the distance between unitary operations $U$ and $V$ can be equivalently transformed into the distance between operations $I$ and $VU^\dg$. Hence, it is of great importance to consider $\cD(I,M)$, where $M$ is a unitary operation. We present some analytical results about the distance between the identity $I$ and some widely-used unitary operations including  generalized controlled phase gate, CNOT, controlled-Z, SWAP gates etc. 

We consider the controlled-phase gate firstly. Let  $U_\t^{(k)}$ be the  two-qubit diagonal operation whose $k$-th diagonal entry is $e^{i\t }$ and other diagonal entries are 1, for $k=1,2,3,4$. We have the following fact.
\begin{proposition}
	\label{pro:D(I,X1)}
	The quantum $W_1$ distance between $I$ and the gate $U_\t^{(3)}=
	\diag\{1,1,e^{i\t},1\}$ is equal to $\sqrt{2}\sin\frac{\t}{2}$, i.e.
	\begin{eqnarray}
		\label{def:D(I,X1)=}
		\cD(I, U_\t^{(3)})
		=&&\max_{\r\in\cS_2}\norm{\r-U_\t^{(3)}\r U_\t^{(3)\dg}}_{W_1}\\
		=&&\sqrt{2}\sin\frac{\t}{2}.
	\end{eqnarray}
\end{proposition}
The proof of Proposition \ref{pro:D(I,X1)} is shown in Appendix \ref{sec:D(I,X)}.

 By applying appropriate local unitary operation, the $U_\t^{(k)}$'s can transform to each other,
\begin{eqnarray}
	(\s_x\otimes I)U_\t^{(3)}	(\s_x\otimes I)=U_\t^{(1)}, \; 
	(\s_x\otimes \s_x)U_\t^{(3)}	(\s_x\otimes \s_x)=U_\t^{(2)},
	\;
	(I\otimes \s_x)U_\t^{(3)}	(I\otimes \s_x)=U_\t^{(4)}.
\end{eqnarray}
Since $\norm{\cdot}_{W_1}$ is invariant under local unitary operation, we obtain the following fact. 
\begin{corollary}
	\label{cor:D(I,UK)}
	 The quantum $W_1$ distance between $I$ and controlled-phase gate $U_\t^{(k)}$ is equal to $\sqrt{2}\sin\frac{\t}{2}$, i.e.
	\begin{eqnarray}
		\label{def:D(I,UK)}
		\cD(I, U_\t^{(k)})=\sqrt{2}\sin\frac{\t}{2},
	\end{eqnarray}
\end{corollary}

The CNOT and controlled-Z gate are most widely-used controlled gate in computation. First we obtain $\cD(I,\rm CZ)$ using Corollary \ref{cor:D(I,UK)}. Then $\cD(I,\rm CNOT)$ is derived by analyzing the relation between $\cD(I,\rm CNOT)$ and $\cD(I,\rm CZ)$. 

Obviously, one has $U_{CZ}=U_\pi^{(4)}$ in  Corollary \ref{cor:D(I,UK)}. By setting $\t=\pi$, the distance $\cD(I,\rm CZ)$ can be obtained.
\begin{proposition}
	\label{pro:D(I,CZ)}
	The quantum $W_1$ distance between $I$ and controlled-Z gate is equal to $\sqrt{2}$, i.e.
	\begin{eqnarray}
		\cD(I,\mbox{CZ})=\sqrt{2}.
	\end{eqnarray}
\end{proposition} 

The CNOT and controlled-Z gate are locally unitary equivalent. So the relation between $\cD(I,\rm CNOT)$ and $\cD(I,\rm CZ)$ can be  derived by the single-qubit unitary invariance of $\norm{\cdot}_{W_1}$.
\begin{lemma}
	\label{le:CZ=CNOT}
	The distance between $I$ and CNOT gate is equal to that between $I$ and Controlled-Z gate. That is to say, 
	\begin{eqnarray}
		\cD(I,\mbox{CNOT})=\cD(I,\mbox{CZ}).
	\end{eqnarray}
\end{lemma} 
\begin{proof}
	It can be proved using the fact that the Controlled-Z gate $U_{CZ}$ can be prepared with the help of a CNOT gate $U_{CN}$ and two Hadamard gates $H$, i.e.,
	\begin{eqnarray}
		U_{CZ}=(I\otimes H)U_{CN}(I\otimes H).
	\end{eqnarray}
	One can show that
	\begin{eqnarray}
		&&\cD(I,\mbox{CNOT})\\
		=&&\max_{\ketbra{\psi}{\psi}\in\cS_2}
		\norm{(I\otimes H)[\ketbra{\psi}{\psi}-U_{CN}\ketbra{\psi}{\psi}U_{CN}](I\otimes H)}_{W_1}\\
		=&&\max_{\ketbra{\eta}{\eta}\in\cS_2}\norm{\ketbra{\eta}{\eta}-U_{CZ}\ketbra{\eta}{\eta}U_{CZ}}_{W_1}\\
		=&&\cD(I,\mbox{CZ})
	\end{eqnarray}
	where $\ket{\eta}=(I\otimes H)\ket{\ps}$ is any two-partite pure state.
\end{proof}

Using Proposition \ref{pro:D(I,CZ)} and Lemma \ref{le:CZ=CNOT}, the $W_1$ distance between $I$ and CNOT gate is obtained as follows.
\begin{proposition}
	\label{pro:D(I,CNOT)}
	The quantum $W_1$ distance between $I$ and CNOT gate is equal to $\sqrt{2}$, i.e.
	\begin{eqnarray}
		\cD(I,\mbox{CNOT})=\sqrt{2}.
	\end{eqnarray}
\end{proposition}
By now, the $W_1$ distance from the identity $I$ and arbitrary two-qubit controlled gates has been  obtained.

The SWAP gate is also a widely used gate in quantum computation. It  accomplishes a useful task, i.e., swapping the states of the two qubits. In quantum circuits, it can be composed by three CNOT gates. We consider $\cD(I,\rm SWAP)$ and the following result is obtained. 
 \begin{proposition}
 	\label{pro:D(I,SWAP)}
The quantum $W_1$ distance between $I$ and SWAP gate is equal to 2, i.e.
\begin{eqnarray}
	\cD(I,\rm SWAP)=2.
\end{eqnarray}
 \end{proposition}

\begin{proof}
	From Property \ref{eq:bounds}, it holds that  $\cD(I,\mbox{SWAP})\leq 2$. On the other hand, we choose the state $\ket{\xi}=\ket{0,1}$. By the definition of $\cD(U,V)$ in (\ref{def:E(U,V)}), we have
	\begin{eqnarray}
		&&\cD(I,\mbox{SWAP})\\
		=&&\max_{\ketbra{\psi}{\psi}\in\cS_2}\norm{\ketbra{\psi}{\psi}-U_{SW}\ketbra{\psi}{\psi}U_{SW}}_{W_1}\\
		\geq&&\norm{\proj{\xi}-U_{SW}\proj{\xi}U_{SW}}_{W_1}\\
		=&&\norm{\proj{0,1}-\proj{1,0}}_{W_1}=2,
	\end{eqnarray}
	where $\norm{\proj{0,1}-\proj{1,0}}_{W_1}=2$ is obtained from the equality in (\ref{eq:projx-projy}). So we have 
	\begin{eqnarray}
		\cD(I,\mbox{SWAP})=2.
	\end{eqnarray}
\end{proof}

Following the idea of the proof of Proposition \ref{pro:D(I,SWAP)}, we obtain a  more general result. 
\begin{proposition}
	\label{pro:permutation matrix}
	Any two-qubit unitary gates switching $\proj{0,1}$ to $\proj{1,0}$, or equivalently $\proj{1,0}$ to $\proj{0,1}$,  have the same quantum Wasserstein distance with the identity $I$,  i.e.,
	\begin{eqnarray}
		\label{eq:D(I,Uk)}
		\cD(I,U_k)=2,
	\end{eqnarray}
	where
	\begin{eqnarray}
		\label{eq:U_k}
		U_1=\bma
		0&*&*&*\\
		0&*&*&*\\
		0&*&*&*\\
		1&0&0&0
		\ema,
		U_2=\bma
		*&0&*&*\\
		*&0&*&*\\
		0&1&0&0\\
		*&0&*&*
		\ema,
		U_3=\bma
		*&*&0&*\\
		0&0&1&0\\
		*&*&0&*\\
		*&*&0&*
		\ema,
		U_4=\bma
		0&0&0&1\\
		*&*&*&0\\
		*&*&*&0\\
		*&*&*&0
		\ema.
	\end{eqnarray}
For any unitary operations $U,V$, the unitary operations $M$ satisfying 
\begin{eqnarray}
	M=(U\otimes V)U_k(U\otimes V)^\dg,
\end{eqnarray}  
show the same distance  with identity $I$, \begin{eqnarray}
	\label{eq:D(I,M)}
	\cD(I,M)=2.
\end{eqnarray}
\end{proposition}
\begin{proof}
Eq. (\ref{eq:D(I,Uk)}) can be obtained by the same way as the proof of Proposition \ref{pro:D(I,SWAP)}.
Recall the property that the quantum Wasserstein distance is invariant with respect to the unitary operations on single qubit.  We have 
\begin{eqnarray}
	\norm{(U\otimes V)\proj{0,1}(U\otimes V)^\dg-(U\otimes V)\proj{1,0}(U\otimes V)^\dg}_{W_1}=2,
\end{eqnarray}
where $U,V$ are single qubit unitary operations.
From Property \ref{eq:leftinvariant} in Proposition \ref{pro:property},  Eq. (\ref{eq:D(I,M)}) is obtained.
\end{proof}
The $W_1$ distance between identity and all order-4 permutation matrices can be derived by Proposition \ref{pro:permutation matrix}. We consider the representation of order-4 permutation group. They are 
\begin{eqnarray}
P_1=I, 
P_2=\bma 
1&0&0&0\\
0&1&0&0\\
0&0&0&1\\
0&0&1&0\\
\ema,
P_3=\bma 
1&0&0&0\\
0&0&1&0\\
0&1&0&0\\
0&0&0&1\\
\ema,
...,
P_{24}=\bma 
0&0&0&1\\
0&0&1&0\\
0&1&0&0\\
1&0&0&0\\
\ema.
\end{eqnarray}
One can verify that 15 permutation matrices in $\{P_1,...,P_{24}\}$ are included in (\ref{eq:U_k}).  By Proposition \ref{pro:permutation matrix}, one can obtain that the $W_1$ distance between every one of them and identity is equal to two. Other nine permutation matrices are listed as follows,
\begin{eqnarray}
H_1=I,H_2=U_{CN}, 
H_3=\bma 
0&1&0&0\\
1&0&0&0\\
0&0&1&0\\
0&0&0&1\\
\ema,
H_4=\bma 
0&1&0&0\\
1&0&0&0\\
0&0&0&1\\
0&0&1&0\\
\ema,
H_5=\bma 
0&0&1&0\\
0&0&0&1\\
1&0&0&0\\
0&1&0&0\\
\ema,
\\
H_6=\bma 
1&0&0&0\\
0&0&0&1\\
0&0&1&0\\
0&1&0&0\\
\ema,
H_7=\bma 
0&0&1&0\\
0&1&0&0\\
1&0&0&0\\
0&0&0&1\\
\ema,
H_8=\bma 
0&1&0&0\\
0&0&0&1\\
1&0&0&0\\
0&0&1&0\\
\ema,
H_9=\bma 
0&0&1&0\\
1&0&0&0\\
0&0&0&1\\
0&1&0&0\\
\ema.
\end{eqnarray}
We analyze the distance for $\cD(I,H_k)$. We have $(\s_x\otimes I) H_3 (\s_x\otimes I)=H_2$, so $\cD(I,H_2)=\cD(I,H_3)=\sqrt{2}$. The fact that $\cD(I,H_4)=\cD(I,H_5)=1$ has been obtained in Proposition \ref{pro: D(I,sx)=k}. Using the fact that $(I\otimes \s_x)H_6(I\otimes \s_x)=H_7$ and $U_{SW}H_6U_{SW}=H_2$, one has $\cD(I,H_2)=\cD(I,H_6)=\cD(I,H_7)=\sqrt{2}$. By  $(I\otimes \s_x)H_8(I\otimes \s_x)=H_9$, one has $\cD(I,H_8)=\cD(I,H_9)$. Using Lemma \ref{le:tensorization}, we find that the lower bound of them is two. Combined with Property \ref{eq:bounds},  we have $\cD(I,H_8)=\cD(I,H_9)=2$. By now, we have obtained the $W_1$ distance between identity and all the order-4 permutation matrices.

\subsubsection{The $W_1$ distance for multi-qudit operations}
\label{sec:case any n}
We show a fact considering the $W_1$ distance between $I$ and a multi-qudit operation. It shows the local discrimination of quantum operations, which is a unique property of the quantum $W_1$ norm between operations.
\begin{proposition}
	\label{pro: D(I,sx)=k}
For a $n$-qudit operation consisted of tensor product of $k$ Pauli gate $X$ and $n-k$ identity $I$, the quantum $W_1$ distance between it and identity $I$ is equal to $k$, i.e.
\begin{eqnarray}
\cD(I^{\otimes n}, I^{\otimes (n-k)}\otimes X^{\otimes k})=k,
\end{eqnarray}
for $k=1,2,...,n$, up to permutations of the qudits.
\end{proposition}
\begin{proof}
We show the claim for $k=1, 2$, and the claim for $k>2$ can be obtained by a similar way.
	
First we show that $\cD(I^{\otimes n}, X\otimes I^{\otimes(n-1)})=1$ up to permutations of the qudits. For any pure states $\r\in\cS_n$ and $\s=(X\otimes I^{\otimes(n-1)})\r(X\otimes I^{\otimes(n-1)})$, it holds that $\tr_1\r=\tr_1\s$, i.e. $\r$ and $\s$ are neighboring states. From Definition \ref{def:w1diatance}, the quantum $W_1$ distance assigns the distance at most one to any couple of neighboring states, so  $\cD(I^{\otimes n}, X\otimes I^{\otimes(n-1)})\leq 1$. On the other hand, $\cD(I^{\otimes n}, X\otimes I^{\otimes(n-1)})=\max_\r\norm{\r-\s}_{W_1}\geq\norm{\proj{00...0}-\proj{10...0}}_{W_1}=1$ by Lemma \ref{le: recovery of classical W1}. Hence, $\cD(I^{\otimes n},X\otimes I^{\otimes(n-1)})=1$. Using the fact that $\norm{\cdot}_{W_1}$ is invariant with respect to permutations of the qudits, $\cD(I^{\otimes n}, I\otimes X\otimes I^{(n-2)})=...=\cD(I^{\otimes n}, I^{\otimes(n-1)}\otimes X)=1$ is obtained.
	
Next we prove that $\cD(I^{\otimes n},  X^{\otimes 2}\otimes I^{\otimes(n-2)})=2$ up to permutations of the qudits. From Definition \ref{def:w1diatance}, it holds that $\cD(I^{\otimes n},  X^{\otimes 2}\otimes I^{\otimes(n-2)})\geq \norm{\proj{000...0}-\proj{110...0}}_{W_1}=2$. By Property \ref{eq:E(U1OU2,V1OV2)}, we have $\cD(I^{\otimes n},  X^{\otimes 2}\otimes I^{\otimes(n-2)})\leq \cD(I^{\otimes n}, X\otimes I^{\otimes(n-1)})+\cD(I^{\otimes n}, I\otimes X\otimes I^{\otimes(n-2)})=2$. The invariance under permutations can also be derived as the above case.	
\end{proof}

\section{Estimation of the closeness between operations in quantum circuit}
\label{Sec:app to qc}
In this section, we show that the $W_1$ distance  between unitary operations plays an important role in estimating the closeness between operations in quantum circuits. A small $\cD(U,V)$ implies that  any measurement performed on the states $U\ket{\ps}$ shows approximately the same measurement statistics as that of $V\ket{\ps}$, so $U$ and $V$ plays almost the same role in quantum circuits.
So the noisy operation simulates the ideal one well  when they become close under the distance.

As we all know, the set of unitary operations is continuous and thus we can never implement an arbitrary unitary operation exactly by a discrete set of gates. We can only approximate the unitary operation with a series of gates.  Let $U$ be the ideal unitary operation that we wish to implement, and $V$ be the unitary operation that is actually implemented under noise. To compare their effects in a quantum circuit, we assume that they are performed on the same state $\ket{\ps}$, where $\ket{\ps}\in \cS_n$ is an arbitrary state. The $W_1$ distance between them characterize how close their measurement outcome will be in terms of POVM. It is realized by deriving an upper bound of the difference in probability between measurement outcomes. 

\begin{proposition}
	\label{th: PU-PV}
Given two operations $U$ and $V$ performed on the same initial state $\ket{\ps}$.	
Let $M_m \geq 0$ be an element in a POVM performed on $U\ket{\ps}$ and $V\ket{\ps}$, with
 $P_U^{(m)}$ and $P_V^{(m)}$ being the probability of obtaining the outcome $m$ in the measurements, respectively. The difference between $P_U^{(m)}$ and $P_V^{(m)}$ is upper bounded by the quantum $W_1$ distance between $U$ and $V$ as
 \begin{eqnarray}
 	\abs{P_U^{(m)}-P_V^{(m)}}\leq 2\l_0(M_m) \cD(U,V),
 \end{eqnarray}
where $\l_0(M_m)\in (0,1]$ is the maximal eigenvalue of $M_m$.
\end{proposition}
\begin{proof}
Since $P_U^{(m)}$ and $P_V^{(m)}$ is the probability of obtaining the measurement outcome $m$, we have
\begin{eqnarray}
\abs{P_U^{(m)}-P_V^{(m)}}=&&\abs{\bra{\ps}U^\dg M_m U\ket{\ps}-\bra{\ps}V^\dg M_m V\ket{\ps}}.
\end{eqnarray}
The POVM operation $M_m$ is positive with the unique positive square root, denoted by $N_m$, i.e., $M_m=N_m N_m^\dg$ and $\sum_m M_m=I$. Hence,
\begin{eqnarray}
\abs{P_U^{(m)}-P_V^{(m)}}=&&\abs{\tr[N_m^\dg(U\proj{\ps}U^\dg)-V\proj{\ps}V^\dg)N_m]}
\\
=&&\left| \sum_k \l_k\Big(N_m^\dg(U\proj{\ps}U^\dg-V\proj{\ps}V^\dg)N_m\Big)\right|
\\
\leq&&\sum_k
\left|\l_k\Big(N_m^\dg(U\proj{\ps}U^\dg-V\proj{\ps}V^\dg)N_m\Big)\right|
\\
=&&\norm{N_m^\dg(U\proj{\ps}U^\dg-V\proj{\ps}V^\dg)N_m}_1,
\end{eqnarray}
where $\l_k(X)$ denotes the $k$-th eigenvalue of the operator $X$, and  $\l_0\geq\l_1\geq...\geq\l_k...$. The third equality comes from $\norm{X}_1=\sum_k\abs{\l_k(X)}$ for normal operators.

Using the fact that $\norm{ABC}_1\leq\norm{A}_{\infty}\norm{B}_1\norm{C}_\infty$ and $\norm{\r-\s}_1\leq2\norm{\r-\s}_{W_1}$, we have
\begin{eqnarray}
&&\norm{N_m^\dg(U\proj{\ps}U^\dg-V\proj{\ps}V^\dg)N_m}_1
\\
\leq&& \norm{N_m^\dg}_{\infty}\norm{U\proj{\ps}U^\dg-V\proj{\ps}V^\dg}_1\norm{N_m}_\infty\\
\leq&&2\norm{N_m^\dg}_{\infty}\norm{N_m}_\infty\norm{U\proj{\ps}U^\dg-V\proj{\ps}V^\dg}_{W_1}
\\
\leq&&2\norm{N_m^\dg}_{\infty}\norm{N_m}_\infty \cD(U,V)\\
=&&2\l_0(M_m) \cD(U,V),
\end{eqnarray}
where the equality holds for
\begin{eqnarray}
\norm{N_m}_\infty=s_0(N_m)=\sqrt{\l_0(M_m)}.
\end{eqnarray}
Here $s_0(X)$ denotes the maximal singular value of operator $X$.
\end{proof}

Proposition \ref{th: PU-PV} shows that if the distance between $U$ and $V$ is small enough, then any POVM performed on the states $U\ket{\ps}$ shows approximately the same measurement statistics as that of $V\ket{\ps}$. The operations $U$ and $V$ plays almost the same role in quantum circuits as their measurement outcomes occur with almost the same probability. So if a kind of noise takes the ideal operations to another one and they are close under the $W_1$ distance, then the noise has little effect on the ideal operation. 
From the perspective of unitary operation discrimination, a small $\cD(U,V)$ also implies that $U$ and $V$ cannot be perfectly distinguished.

We have characterized the distance between  individual gates in Proposition \ref{th: PU-PV}. In the quantum circuits, the realization of target operations always requires a sequence of unitary gates.  So it is important to obtain the distance between two sequences of gates.  In analogy to quantifying the distance between an entangled state and a product state, one may be interested in the distance    between a nonlocal quantum gate and the tensor product gate. 
\begin{proposition}
	\label{th: EUm}
Two sequences of multi-qubit unitary gates $U_tU_{t-1}...U_1$ and $V_{t}V_{t-1}...V_1$ acting on the state space $\cS_n$, where where $V_j$ can be decomposed as the tensor product of single-qubit gates, for $j=1,2,...,t$.  The quantum $W_1$ distance between them adds at most linearly with respect to the distance of each couple of gates,
\begin{eqnarray}
	\cD(U_tU_{t-1}...U_1, V_{t}V_{t-1}...V_1)
	\leq
	\sum_{k=1}^t \cD(U_k,V_k).
\end{eqnarray} 
\end{proposition}
\begin{proof}
We prove it by induction. First we show the case for $t=2$.
\begin{eqnarray}
&&\cD(U_2U_1,V_2V_1)\\
=&&\max_\r\norm{U_2U_1\r U_1^\dg U_2^\dg-V_2V_1\r V_1^\dg V_2^\dg}_{W_1}
\\
\leq&&\max_\r
\norm{U_2 (U_1\r U_1^\dg) U_2^\dg-
	V_2(U_1\r U_1^\dg) V_2^\dg}_{W_1}
+\max_\r
\norm{	V_2U_1\r U_1^\dg V_2^\dg-	
V_2V_1\r V_1^\dg V_2^\dg}_{W_1}
\\
= && \cD(U_2,V_2)+\cD(V_2U_1,V_2V_1).
\end{eqnarray}
Using property \ref{eq:leftinvariant} of $\cD(U,V)$ in Proposition \ref{pro:property}, one has
\begin{eqnarray}
\cD(V_2U_1,V_2V_1)=\cD(U_1,V_1).
\end{eqnarray}
So we have
\begin{eqnarray}
\cD(U_2U_1,V_2V_1)\leq
\cD(U_2,V_2)+\cD(U_1,V_1).
\end{eqnarray}
Suppose the case for $t-1$ holds, i.e., $\cD(U_{t-1}...U_1, V_{t-1}...V_1)
\leq
\sum_{k=1}^{t-1} \cD(U_k,V_k)$. Then we have 
\begin{eqnarray}
&&\cD(U_t U_{t-1}...U_1, V_{t}V_{t-1}...V_1)
\\
\leq&&
\cD(U_t,V_t)+\cD(U_{t-1}U_{t-2}...U_1, V_{t-1}V_{t-2}...V_1)
\\
\leq&&\sum_{k=1}^t \cD(U_k,V_k),
\end{eqnarray}
which is the desired result.
\end{proof}
 The above fact characterizes the distance between two sequences of gates. One sequence of gates consists of the gates that can be decomposed as the tensor product of single-qubit gates. The other one consists of arbitrary multi-qubit gates.  It shows that the distance of the entire sequence of gates is at most the sum of the distance of individual  gates.

Proposition \ref{th: PU-PV} and \ref{th: EUm} can be applied to estimate the measurement outcome of the circuits containing different sequence of gates $U_1,U_2,...,U_t$ and $V_1,V_2,...,V_t$. In practice, we set a tolerance $\a>0$ of the probability that two circuits show the same measurement outcome.
We can estimate how close the effects of these gates are in the circuits, i.e., whether the probability of different measurement outcomes are within the tolerance, only by the distance $\cD(U_k,V_k)$. To be specific, to make the probability of different measurement outcomes be within the tolerance $\a$, it suffices that
\begin{eqnarray}
	\label{ieq:PUt...U1}
\abs{P_{U_{t}...U_1}^{(m)}-P^{(m)}_{V_{t}...V_1}}
\leq 
2\l_0(M_m)\sum_{k=1}^t \cD(U_k,V_k)
\leq\a,
\end{eqnarray}
where $P_{U_{t}...U_1}^{(m)},P^{(m)}_{V_{t}...V_1}$ and $\l_0(M_m)$ follows the symbolic hypothesis in Proposition \ref{th: PU-PV}.
The inequality (\ref{ieq:PUt...U1}) holds when
\begin{eqnarray}
 \cD(U_k,V_k)
\leq \frac{\a}{2t\max_m\{\l_0(M_m)\}}.
\end{eqnarray}
It is shown in FIG. \ref{fig:eg1}.

\begin{figure}[htp]
	\center{\includegraphics[width=12cm]  {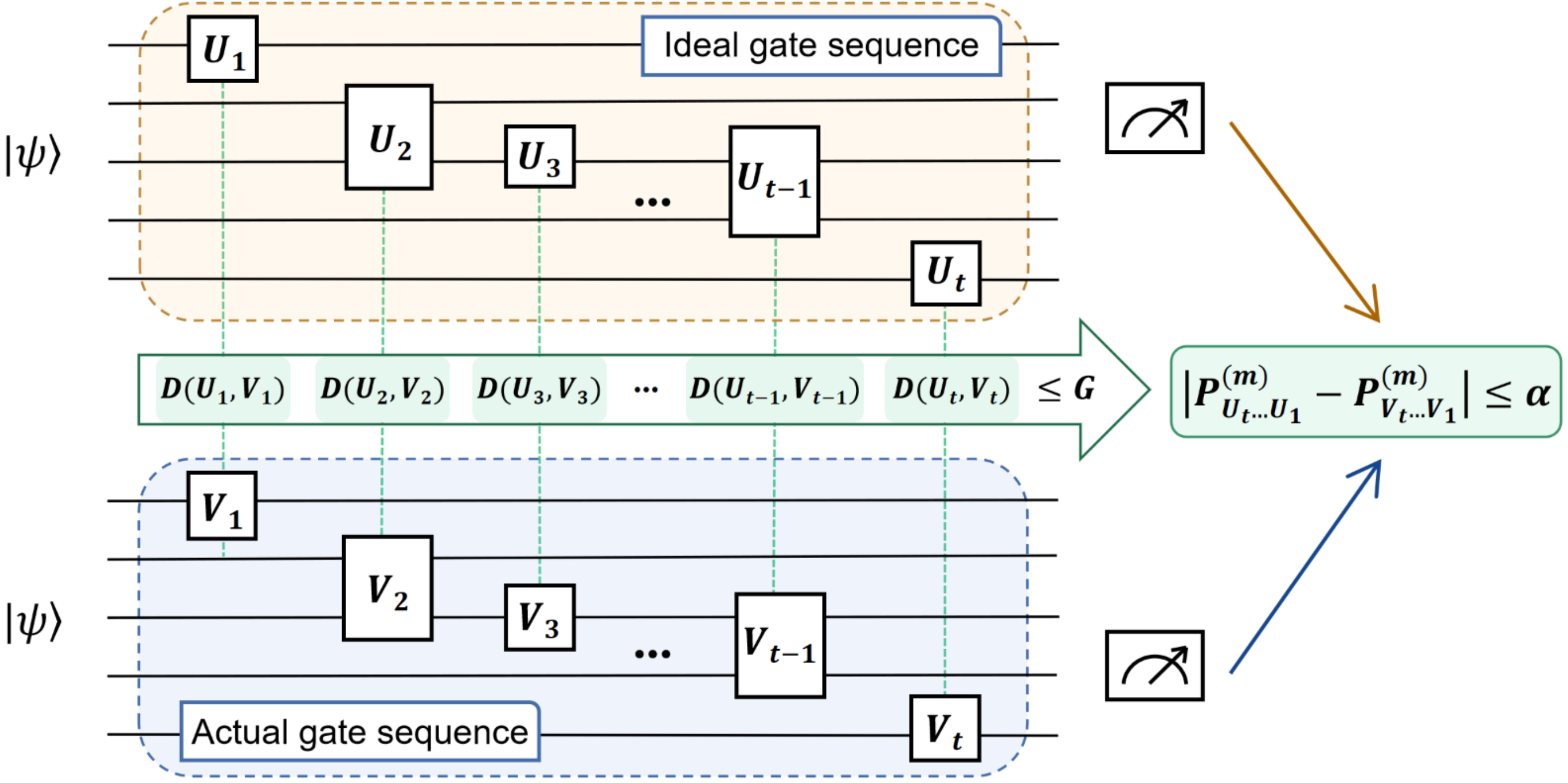}}
	\caption{Diagram showing the evaluation for the measurement outcome of the circuits containing ideal gate sequence $U_1,U_2,...,U_t$ and actual gate sequence $V_1,V_2,...,V_t$. 
	The two sequence of gates are performed on the same initial state $\ket{\ps}$.	
	Let $M_m \geq 0$ be an element in an arbitrary POVM, with
	$P_{U_t...U_1}^{(m)}$ and $P_{V_t...V_1}^{(m)}$ being the probability of obtaining the outcome $m$ in the measurements, respectively.  We set a tolerance $\a>0$ of the probability that two circuits show the same measurement outcome.	To make $\abs{P_{U_{t}...U_1}^{(m)}-P^{(m)}_{V_{t}...V_1}}	\leq \a$, it suffices to guarantee $D(U_k,V_k)\leq G= \a/(2t\max_m\{\l_0(M_m)\})$. }
	\label{fig:eg1}
\end{figure}
Now we show an example of the above process.
\begin{example}
A sequence of ideal qubit gates $U_1,U_2,...,U_5$ in the quantum circuit is subject to the unitary noise process $\cE_\t=\diag\{e^{i\t},e^{-i\t}\}$, where $\t\in[0,\pi]$ is the parameter related to noise.  The ideal gates are transformed into a sequence of noisy gates $V_1,V_2,...,V_5$, where
\begin{eqnarray}
U_k=\bma
e^{\a_k i}&0\\
0&e^{\b_k i}
\ema, \quad
V_k=U_k\cE_\t=\bma
e^{(\a_k+\t)i}&0\\
0&e^{(\b_k-\t)i}
\ema,
\end{eqnarray}  
 for  $k=1,2,...,5$.
Using the results for the calculation of $\cD(U,V)$ in Sec. \ref{sec:case n=1}, one has 
\begin{eqnarray}
\cD(U_k,V_k)=\abs{\sin\t}.
\end{eqnarray}
Suppose the following POVM  $\{M_m:m=1,2,...,8\}$ is carried out in the circuit,
\begin{eqnarray}
&&M_1=\bma
\frac{1}{8}&-\frac{1}{8}i\\
\frac{1}{8}i&\frac{1}{8}
\ema,
\quad
M_2=\bma
\frac{1}{8}&\frac{1}{8}i\\
-\frac{1}{8}i&\frac{1}{8}
\ema,
\quad
M_3=\bma
\frac{1}{8}&\frac{1}{8}\\
\frac{1}{8}&\frac{1}{8}
\ema,
\quad
M_4=\bma
\frac{1}{8}&-\frac{1}{8}\\
-\frac{1}{8}&\frac{1}{8}
\ema,\\
&&M_5=\bma
\frac{1}{8}&\frac{1}{8}e^{\pi i/4}\\
\frac{1}{8}e^{-\pi i/4}&\frac{1}{8}
\ema,
M_6=\bma
\frac{1}{8}&\frac{1}{8}e^{5\pi i/4}\\
\frac{1}{8}e^{-5\pi i/4}&\frac{1}{8}
\ema,
M_7=\bma
\frac{1}{4}&0\\
0&0
\ema,
M_8=\bma
0&0\\
0&\frac{1}{4}
\ema.
\end{eqnarray} 
We set the probability error tolerance $\a=30\%$. To make the probability of different measurement outcomes be within the tolerance $\a$ for any initial state $\ket{\ps}$, i.e. 	$\abs{P_{U_{t}...U_1}^{(m)}-P^{(m)}_{V_{t}...V_1}}\leq \a$,  it suffices that
\begin{eqnarray}
	\cD(U_k,V_k)=\abs{\sin\t}
	\leq \frac{\a}{10\max_m\{\l_0(M_m)\}}=0.12.
\end{eqnarray}
It is the sufficient condition for each couple of gates concerning only the noise. It implies that each noisy gate $V_k$ simulates the ideal gate $U_k$   within the tolerance effectively if the parameter of local noise $\t\in[0, \arcsin(0.12)]$ or $\t\in[\pi-\arcsin(0.12),\pi]$.
\end{example}

\section{The $W_1$ gate error rate under noise} 
\label{sec:gate error rate}
In this section, we introduce an measure of the realization of quantum gates, named  $W_1$ gate error rate. We show its rationality compared with the gate error rate induced by other norm, and estimate the $W_1$ gate error rate with the help of $W_1$ distance between operations. On the basis of that, we establish the relation between the $W_1$ error rate and two real cost measure of recover operation, including circuit cost and experiment cost. The $W_1$ error rate is related to the practical cost of eliminating the effect of noise on a specific type of gate, i.e., the low $W_1$ error rate of a gate implies that it will cost less to eliminate the effect of noise on it.  Further we show two examples considering the implementation under depolarizing and unitary noise  for arbitrary single-qubit gate and CNOT gate, respectively.

Following the idea of proposing gate error rate in Sec. \ref{sec:preliminary error rate}, we define the error rate of $U $ by the quantum $W_1$ distance as follows.
\begin{definition}
	\label{def:errorrate}
The $W_1$ error rate of the implementation of $n$-qubit unitary gate $U$ is given by
\begin{eqnarray}
e(U,\cV):=\frac{1}{n}\max_{\r}\norm{U\r U^\dg-\cV(\r)}_{W_1},
\end{eqnarray}
where $\r\in \cS_n$ and $\cV$ is a channel that describes the noisy implementation of $U$.
\end{definition}
For any states $\r,\s\in\cS_n$, we have $\norm{\r-\s}_{W_1} \in[0,n]$ by Definition \ref{def:w1diatance} and hence the error rate $ e(U)\in[0,1]$. Compared with the error rate induced by Schatten 1-norm in (\ref{def:e1(u,v)}) and diamond norm in (\ref{def:ed(u,v)}), the following relation can be obtained
\begin{eqnarray}
e(U,\cV)\leq e_1(U,\cV) \leq e_\diamond(U,\cV),
\end{eqnarray} 	
where the first inequality comes form Lemma \ref{le:relation with trace norm}, and the second one follows directly from their definitions. 

 As we all know, quantum error correction (QEC) is a two stage process: the error detection step, followed by the recovery step using conditioned unitary operations \cite{nielsen2010quantum}. We denote the operations used in the second step as recovery operations, which are performed to eliminate the influence of noise on specific qubit. Compared with the former error rate induced by other distance like fidelity, Schatten 1-norm, and diamond norm, the error rate $e(U)$ has a better explanation from the perspective of experiment cost for the recovery operations. It comes from the property of $W_1$ norm. From (\ref{eq:r'Or''}) in Lemma \ref{le:tensorization}, operations which reduce the distance between two states over a portion of their qubits will proportionally reduce the total distance over all of the qubits, while no unitarily invariant distance have this property \cite{de2021the,bobak2022learning}. For example, an ideal gate $U_{id}=\s_x\otimes\s_x$ is performed on $\r=\proj{00}$ to generate $\r_{id}$. Two noisy implementation of $U_{id}$ shows $U_{ac}^{(1)}=I^{\otimes2}$ and $U_{ac}^{(2)}=I\otimes\s_x$, whose  resulting states are $\r_{ac}^{(1)}$ and $\r_{ac}^{(2)}$,  respectively. Since $\r_{ac}^{(k)}$ are orthogonal to the ideal state, all the distance induced by  unitary-invariant norms and fidelity shows that $\norm{\r_{id}-\r_{ac}^{(1)}}=\norm{\r_{id}-\r_{ac}^{(2)}}$.  Using quantum $W_1$ norm, we have $\norm{\r_{id}-\r_{ac}^{(1)}}_{W_1}=2>1=\norm{\r_{id}-\r_{ac}^{(2)}}_{W_1}$. In the recovery step of QEC, two gates $\s_x\otimes\s_x$ are required for $\r_{ac}^{(1)}$, and only one gate $\s_x$ on the first qubit  is required for $\r_{ac}^{(2)}$.  So $\r_{ac}^{(1)}$ is further away from $\r_{id}$ than $\r_{ac}^{(2)}$ in terms of experiment resource, which is consistent with the distance induced by $W_1$ norm.

We demonstrate Definition \ref{def:errorrate} as follows. Consider the noise process described by mixed unitary channel $\cV=\cG\circ\cE$. Here $\cG(\cdot)=U(\cdot)U^\dg$ denotes the ideal implementation of gate $U\in\cU_n$ and $\cE$ is the channel describes the effect of noise.
First we analyze a general noise process described by the generalized quantum operations comprising finite linear combinations of unitary quantum operations (also called mixed unitary channel) \cite{girard2022the},
\begin{eqnarray}
	\label{eq:cE}
	\cE(\cdot)=\sum_{k=1}^N p_kV_k(\cdot)V_k^\dg,
\end{eqnarray}
where $(p_1,p_2,...,p_N)$ is a probability vector and $V_1,V_2,...,V_N\in\cU_n$. Such a channel is considered as many natural examples of noisy channels including the dephasing and depolarizing channels are mixed unitary \cite{burrell2009geometry}.
In the presence of this noise, the error rate of gate $U$ is
\begin{eqnarray}
	e(U, \cV)=\frac{1}{n}\max_\r\norm{U\r U^\dg-U(\sum_{k=1}^N p_kV_k\r V_k^\dg)U^\dg}_{W_1}.
\end{eqnarray}
Since calculating $	e(U, \cV)$ directly is not an easy task, we can derive its upper bound as follows
\begin{eqnarray}
		e(U, \cV)\leq&&\sum_kp_k\left\{\frac{1}{n}\max_\r\norm{U\r U^\dg-UV_k\r V_k^\dg U^\dg}_{W_1}\right\}
	\\
	\label{ieq:e(u,v)D(u,vk)}
	=&&\sum_kp_k\left\{\frac{1}{n}\cD(I,UV_kU^\dg)\right\},
\end{eqnarray}
where the inequality is derived from the convexity of $\norm{\cdot}_{W_1}$. The upper bound in (\ref{ieq:e(u,v)D(u,vk)}) can be used to establish the relation between $W_1$ error rate and its recovery operation.

To analyze the upper bound of $	e(U, \cV)$, it suffices to consider the noise process described by each unitary error $V_k$. For convenience, we denote $E=V_k \in\cU_n$ is an arbitrary element in (\ref{eq:cE}).   It is a unitary operation with eigenvalues $e^{i\t_j}$, where $\t_j\in[0,2\pi)$ and $j=1,2,...,2^n$. The channel describing this unitary error and the corresponding noisy implementation of $U$  are respectively
\begin{eqnarray}
	\label{eq:cE=E}
&&\cE_E(\cdot)=E (\cdot)E^\dg,\\
\label{eq:cE=G}
&&\cV_E=\cG\circ \cE_E,
\end{eqnarray}
where $\cG(\cdot)=U(\cdot)U^\dg$ denotes the ideal implementation of gate $U$.
 The error rate of gate $U$ under unitary noise is 
\begin{eqnarray}
\label{eq:etU}
e(U,\cV_E)=&&
\frac{1}{n}\max_{\r}\norm{U\r U^\dg-UE\r E^\dg U^\dg}_{W_1}=\frac{1}{n}\cD(I,P^E),
\end{eqnarray}
where 
\begin{eqnarray}
P^E=U E^\dg U^\dg
\end{eqnarray}
is the recovery operation of ideal gate $U$ in the presence of unitary error described by $E$. Note that performing $P^E$ on the noisy gate $UE$ can correct the influence of noise, i.e. $P^E(UE)=U$. The second equality in (\ref{eq:etU}) comes from Properties \ref{eq:rightinvariant} and \ref{eq:E(I,Udg)} in Proposition \ref{pro:property}.

Now we establish the relation between the error rate $e^E(U)$ and the experiment cost of recover operation. The circuit complexity of a unitary operation is defined as the minimal number of basic gates needed to generate this operation \cite{nielsen2010quantum}. Circuit cost of quantum circuits, is then proposed to be a lower bound for the circuit complexity \cite{nielsen2006geometry,nielsen2006bouns}. Experiment cost $\cR(U)$, showing  quantum limit on converting quantum resources including energy and time to computational resources,  is also an important complexity measure \cite{Girolami2021Quantifying}.  Recently, the lower bounds for circuit cost and experiment cost are obtained in terms of the quantum Wasserstein complexity measure \cite{2208.06306}. We rephrase their results by our quantum $W_1$ distance between unitary operations. That is,
\begin{eqnarray}
	\label{ieq:CUgeq}
\cC(U)\geq&&4\sqrt{2}\cD(I,U),
\\
	\label{ieq:RUgeq}
	\cR(U)\geq&&\frac{1}{2}\cD(I,U).
\end{eqnarray}

Using (\ref{eq:etU})-(\ref{ieq:RUgeq}), we can obtain that
\begin{eqnarray}
		\label{ieq:CUleq}
e(U, \cV_E)\leq && \frac{1}{4\sqrt{2}n}\cC(P^E),
\\
	\label{ieq:RUleq}
e(U, \cV_E)\leq &&\frac{2}{n}\cR(P^E),
\end{eqnarray}
 Eqs. (\ref{ieq:CUleq}) and (\ref{ieq:RUleq}) imply that the $W_1$ error rate $e(U, \cV_E)$ provides a lower bound for circuit and experiment cost to realize the recovery operation under the unitary noise described by $\cE_E$. That is to say, $e(U, \cV_E)$ is related to quantum resources required to  eliminate the influence of noise $\cE_E$ during the implementation of $U$. Recall that in (\ref{ieq:e(u,v)D(u,vk)}), the error rate of a mixed unitary channel $e(U,\cV)$ can be upper bounded by convex sum of the error rate of each Kraus operator, i.e.
 \begin{eqnarray}
e(U,\cV)\leq \sum_k p_ke(U,\cV_{k}),
 \end{eqnarray} 
 where $\cV_k$ is defined by setting $E=V_k$ in (\ref{eq:cE=E}) and (\ref{eq:cE=G}). Using (\ref{ieq:CUleq}) and (\ref{ieq:RUleq}), we have
 \begin{eqnarray}
 	\label{ieq:e(U,CV),C}
e(U,\cV)\leq&&\sum_k\frac{p_k}{4\sqrt{2}n}\cC(P^{V_k}),
 	\\
 	\label{ieq:e(U,CV),R}
e(U,\cV)\leq&&\sum_k\frac{2p_k}{n}\cR(P^{V_k}). 	
 \end{eqnarray}
So the lower bound of  circuit and experiment cost for the recover operation under arbitrary noise process is obtained. Thus the $W_1$ error rate is a new figure of merit concerning the noisy gate and  the experimental requirement to eliminate the influence of noise on it.

\begin{example}
We consider the depolarizing noise and unitary noise acting on a single qubit. The noise process is given by the channels respectively,
\begin{eqnarray}
&&\cE^{dep,1}_p(\r):=(1-p)\r+p\frac{\bbI_2}{2},
\\
&&\cE^{uni}_\t(\r):=E_\t\r E_\t^\dg,
\end{eqnarray}
where $p\in [0,1]$ and $E_\t$ is a unitary operator with eigenvalues $e^{\pm \t i}$ for $\t\in[0,\pi]$.
The error rate of a single-qubit gate $U$ is
\begin{eqnarray}
e(U, \cV_{dep,1})=&&\max_{\r}\norm{U\r U^\dg-[(1-p)U\r U^\dg+\frac{p}{2}\bbI_2]}_{W_1}=\frac{p}{2},
\\
e(U,\cV_{uni})=&&\max_{\r}\norm{U\r U^\dg-UE_\t\r E_\t^\dg U^\dg}_{W_1}=\sqrt{1-\cos2\t}.
\end{eqnarray}
The average gate fidelity for depolarizing noise and unitary noise is respectively \cite{sanders2015bounding},
\begin{eqnarray}
\varphi^{dep}_p=1-\frac{p}{2},\quad
\varphi^{u}_\t=\frac{1}{3}+\frac{2}{3}\cos^2\t.
\end{eqnarray}
The error rate induced by the diamond norm is \cite{sanders2015bounding}
\begin{eqnarray}
e_{\diamond}(U,\cV_{dep,1})=\frac{3}{4}p, \quad
e_{\diamond}(U,\cV_{uni})=\sin\t. 
\end{eqnarray}
\end{example}

Generally, the advantage of quantum $W_1$ norm appears for multi-qubit case. We show the example for the $W_1$ error rate of noisy implementation of CNOT gate in the presence of two typical kinds of noise.
\begin{example}
	\label{ex:twoqubit}
We consider the noisy implementation of CNOT gate on $\cS_2$ under unitary noise and depolarizing noise  as follows.
\begin{enumerate}
	\item 
	We consider the $W_1$ error rate of noisy implementation of CNOT gate under the following unitary noise channel
	\begin{eqnarray}
		\cE_{CP}(\r)=U_{CP}\r U_{CP}^\dg, \; \mbox{for}\; U_{CP}=\diag\{1,1,1,e^{i\t }\},\;\t\in[0,2\pi).
	\end{eqnarray}
	We denote the actual implementation of CNOT gate in the presence of unitary noise as $\cV_{uni,2}= \cG_{CN}\circ \cE_{CP}$. From  Definition \ref{def:errorrate}, it can be given as 
	\begin{eqnarray}
		\label{eq:e(cnot uni2)}
		e(\mbox{CNOT},\cV_{uni,2})=\frac{1}{2}\cD(I,U_{CN}U_{CP}U_{CN}).
	\end{eqnarray}
	By some calculations shown in Appendix \ref{sec:D(I,X)},  the $W_1$ error rate  of CNOT gate under unitary noise is 
	\begin{eqnarray}
		e(\mbox{CNOT},\cV_{uni,2})=\frac{1}{\sqrt{2}}\sin\frac{\t}{2}.
	\end{eqnarray}
From (\ref{ieq:CUleq}) and (\ref{ieq:RUleq}), the lower bounds for circuit cost and experiment cost are  $\cC(U_{CN}U_{CP}U_{CN})=4\sqrt{2}ne(\mbox{CNOT},\cV_{uni,2})=8\sin\frac{\t}{2}$ and $\cR(U_{CN}U_{CP}U_{CN})=\frac{n}{2}e(\mbox{CNOT},\cV_{uni,2})=\frac{1}{\sqrt{2}}\sin\frac{\t}{2}$, respectively, which is the
quantum resource required to  eliminate the influence of noise $\cV_{uni,2}$ during the implementation of CNOT gate.
	\item The depolarizing channel acting on $\cS_2$ is
\begin{eqnarray}
	\label{def:twodepolarizing}
	\cE^{dep,2}_p(\r):=&&(1-p)\r+p\frac{\bbI_4}{4}\\
	=&&(1-p)\r+\frac{p}{16}\sum_{s,t=0}^3(X^sZ^t)(\r)(X^sZ^t)^\dg,
\end{eqnarray}
where $X=\sum_{q=0}^3\ket{q\oplus 1}\bra{q}, Z=\sum_{q=0}^3 i^q\ket{q}\bra{q}$. We denote 
$\cV_{dep,2}= \cG_{CN}\circ \cE^{dep,2}_p$
as the actual implementation of CNOT gate in the presence of depolarizing noise, for $\cG_{CN}(\cdot)=U_{CN}(\cdot)U_{CN}$. The $W_1$ error rate of the CNOT gate can be estimated as follows
\begin{eqnarray}
	e(\mbox{CNOT}, \cV_{dep,2})=
	\frac{1}{2}\max_\r \norm{\cG_{CN}(\r)-\cG_{CN}\circ\cE^{dep,2}_p(\r)}_{W_1}\in[\frac{3}{8}p,\frac{3}{4}p],
\end{eqnarray}
where the range  is derived from Lemma \ref{le:relation with trace norm}, i.e., $\frac{3}{8}p=\frac{p}{2}\max_\r\frac{1}{2}\norm{U_{CN}\r U_{CN}-\frac{\bbI_4}{4}}_{1}\leq	e(\mbox{CNOT}, \cV_{dep,2})
\leq\frac{p}{2}\max_\r\norm{U_{CN}\r U_{CN}-\frac{\bbI_4}{4}}_{1}=\frac{3}{4}p
$.
Using (\ref{ieq:e(u,v)D(u,vk)}) and  (\ref{def:twodepolarizing}), we have
\begin{eqnarray}
	e(\mbox{CNOT}, \cV_{dep,2})\leq\frac{p}{32}\sum_{s,t=0}^3\cD(I,U_{CN}X^sZ^tU_{CN}).
\end{eqnarray}
From (\ref{ieq:e(U,CV),C}) and (\ref{ieq:e(U,CV),R}), the average lower bounds of circuit cost and experiment cost with respect to the depolarizing noise are $\frac{p}{16}\sum_{s,t\neq 0}\cC(U_{CN}X^sZ^tU_{CN})\geq3\sqrt{2}p$ and $\frac{p}{16}\sum_{s,t\neq 0}\cR(U_{CN}X^sZ^tU_{CN})\geq\frac{3}{8}p$, respectively. 
\end{enumerate}

\end{example}

\section{Conclusion}
\label{Sec:conclusion}
In summary, we have introduced the quantum Wasserstein distance between unitary operations, which characterizes the local distinguishability of operations. We presented its properties and showed its analytical calculation. The closeness between operations can be estimated in quantum circuits with the quantum Wasserstein distance between operations. The smaller the distance between two operations is, the similar their measurement outcome will be in the circuit.
As an application, we introduced the $W_1$ error rate by the distance. We showed its estimation by the quantum Wasserstein distance between operations,  and established the relation between the $W_1$ error rate and two real cost measure of recover operation, including circuit cost and experiment cost. We showed two examples
considering the implementation under depolarizing and unitary noise for arbitrary single-qubit gate
and CNOT gate, respectively.

Many problems arising from this paper can be further explored. The calculation of the quantum Wasserstein distance between operations requires taking the maximization over all states. The efficient approximation of that may be constructed by sampling method in \cite{Yiyou20221quantum}, and its optimal design of sampling circuit  may be given. As an similarity measure of operations, the quantum Wasserstein distance between operations may be employed to design the loss functions in quantum operation learning and make the learning more efficient. Besides, we have shown that $\cD(I,\rm CNOT)=\sqrt{2}$ and $\cD(I,\rm SWAP)=2$. In \cite{chen2016entangling}, it has been obtained that the entangling power of CNOT and SWAP gates are one and two ebits, respectively. So our distance may be developed to characterize more properties of unitary operations including the entangling power.

\section*{ACKNOWLEDGMENTS}
 The authors were supported by the NNSF of China (Grant No. 11871089) and the Fundamental Research Funds for the Central Universities (Grant No. ZG216S2005).

\appendix

\section{The $W_1$ distance between $I$ and generalized controlled phase gate}
\label{sec:D(I,X)}
We present the calculation of the distance between $I$ and generalized controlled phase gate $U_\t^{(k)}$, which is the diagonal two-qubit operation whose $k$-th diagonal entry is $e^{i\t }$ and other diagonal entries are 1 for $k=1,2,3,4$.

 First we consider $\cD(I,U_\t^{(3)})$, as  $U_\t^{(3)}=U_{CN} U_{CP} U_{CN}$ will be used  in
  (\ref{eq:e(cnot uni2)}) from Example \ref{ex:twoqubit}, where $U_{CP}$ and $U_{CN}$ are given in Sec. \ref{sec:notation}.
We rephrase Proposition \ref{pro:D(I,X1)} from Sec. \ref{sec:calculation} for convenience. 
\begin{proposition}
	\label{pro:D(I,X2)}
	The quantum $W_1$ distance between $I$ and controlled-phase gate $U_\t^{(3)}=U_{CN} U_{CP} U_{CN}=
	\diag\{1,1,e^{i\t},1\}$ is equal to $\sqrt{2}\sin\frac{\t}{2}$, i.e.
	\begin{eqnarray}
		\label{def:D(I,X)=}
		\cD(I, U_{CN} U_{CP} U_{CN})
		=&&\max_{\r\in\cS_2}\norm{\r-(U_{CN} U_{CP} U_{CN})\r(U_{CN} U_{CP} U_{CN})^\dg}_{W_1}\\
		=&&\sqrt{2}\sin\frac{\t}{2}.
	\end{eqnarray}
\end{proposition}

\begin{proof}  We calculate $\cD(I,U_{CN} U_{CP} U_{CN})$ by finding its upper and lower bounds. If the upper bound coincide with the lower bound, then we obtain the desired value.
 
 For a pure state $\ket{\ps}=\sum_{m,n=0}^1a_{m,n}\ket{m,n}$ with $\sum_{m,n}\abs{a_{m,n}}^2=1$, we set the state $\r=\proj{\ps}$ and $\s=(U_{CN}U_{CP}U_{CN})\r(U_{CN} U_{CP}  U_{CN})^\dg$. 
From Corollary \ref{cor:geqnorm1}, it can be obtained  that 
\begin{eqnarray}
	\label{ieq:Xnormr1-s1,1}
	\norm{\r-\s}_{W_1}
	\geq &&
	\frac{1}{2}\norm{\r_1-\s_1}_1+\frac{1}{2}\norm{\r_2-\s_2}_1
	\\
		\label{ieq:Xnormr1-s1,2}
	=&&2\sin\frac{\t}{2}\abs{a_{1,0}}(\abs{a_{0,0}}+\abs{a_{1,1}}),
\end{eqnarray}
where $\r_k$ and $\s_k$ are the reduced density operator of $\r$ and $\s$, respectively.
From $\sum_{m,n}\abs{a_{m,n}}^2=1$, we have $2\abs{a_{1,0}}(\abs{a_{0,0}}+\abs{a_{1,1}})\leq \sqrt{2}$, and the equality holds for the input state $\r_{inf}=\proj{\ps_1}$ and $\s_{inf}=(U_{CN} U_{CP}  U_{CN})\r_{inf}(U_{CN}U_{CP}U_{CN})^\dg$, where $\ket{\ps_1}=\sum_{m,n}a_{m,n}\ket{m,n}$ with the coefficients  satisfying $\abs{a_{0,0}}=\abs{a_{1,1}}=\frac{1}{2}$, $\abs{a_{0,1}}=0$, $\abs{a_{1,0}}=\frac{1}{\sqrt{2}}$. From (\ref{def:D(I,X)=}), $\cD(I,U_{CN}U_{CP}U_{CN})$ is obtained by taking the maximization over all input states. From (\ref{ieq:Xnormr1-s1,1}) and (\ref{ieq:Xnormr1-s1,2}), we have obtained that  $\norm{\r_{inf}-\s_{inf}}_{W_1}\geq\sqrt{2}\sin\frac{\t}{2}$. Hence we have
\begin{eqnarray}
	\label{ieq: rangeX}
	\cD(I,U_{CN}U_{CP}U_{CN})\geq 	\sqrt{2}\sin\frac{\t}{2}.
\end{eqnarray}  
According to Definition \ref{def:w1diatance},  one has
\begin{eqnarray}
	\label{eq:E(I,X)def}
	&&\cD(I,U_{CN}U_{CP}U_{CN})\\
	=&&\max_{\r\in\cS_2}\norm{\r-(U_{CN}U_{CP}U_{CN})\r (U_{CN}U_{CP}U_{CN})^\dg}_{W_1}
	\nonumber\\
	=&&\max_{\r\in\cS_2}\min \left\{\sum_{i=1}^2 c_i: c_i\geq0, \r-(U_{CN}U_{CP}U_{CN})\r (U_{CN}U_{CP}U_{CN})^\dg=\sum_{i=1}^2c_iF^{(i)},\tr_i F^{(i)}=0\right\},\nonumber
\end{eqnarray}
where $F^{(i)}=\r^{(i)}-\s^{(i)}\in \cM_2$ satisfying $\tr_iF^{(i)}=0$. Since $\cD(I,U_{CN}U_{CP}U_{CN})$ is derived by taking the minimization of $c_1+c_2$ over all $F^{(i)}$'s, the coefficient $c_1+c_2$ induced by a particular set of $F^{(1)}$ and $F^{(2)}$ is the upper bound of  $\cD(I,U_{CN}U_{CP}U_{CN})$.
We consider the particular case for  $F^{(i)}=\r^{(i)}-M\r^{(i)}M\in \cM_2$  satisfying $\tr_iF^{(i)}=0$, for $M=\diag\{1,1,-1,1\}$. Here $\r^{(i)}$ is any two-qubit pure state.  We aim to show that the upper bound of (\ref{eq:E(I,X)def}) is equal to $\sqrt{2}\sin\frac{\t}{2}$. It can be realized by finding a couple of $F^{(1)}$ and $F^{(2)}$ such that $c_1+c_2\leq\sqrt{2}\sin\frac{\t}{2}$  for all  pure states $\r$.  

Generally, for any pure state $\r=
\sum_{j,k,s,t}a_{j,k}a_{s,t}^*\ketbra{j,k}{s,t}$, one has
\begin{eqnarray}
	\label{eq:r-XrX}
	&&\r-(U_{CN}U_{CP}U_{CN})\r (U_{CN}U_{CP}U_{CN})^\dg
	\nonumber\\
	=&&
	\bma
	0&0&a_{0,0}a_{1,0}^*(1-e^{-i\t })&0\\
	0&0&a_{0,1}a_{1,0}^*(1-e^{-i\t })&0\\
	a_{0,0}^*a_{1,0}(1-e^{ i\t })&a_{0,1}^*a_{1,0}(1-e^{i\t })&0&a_{1,0}a_{1,1}^*(1-e^{i\t })\\
	0&0&a_{1,0}^*a_{1,1}(1-e^{-i\t })&0\\
	\ema.
\end{eqnarray}
Any $F^{(k)}\in \cM_2$ satisfying $\tr_kF^{(k)}=0$ can be written as
\begin{eqnarray}
	\label{eq:XD1D2}
	F^{(1)}=
	\bma
	0&0&2g_{0,0}g_{1,0}^*&0\\
	0&0&2g_{0,1}g_{1,0}^*&0\\
	2g_{0,0}^*g_{1,0}&2g_{0,1}^*g_{1,0}&0&0\\
	0&0&0&0\\
	\ema,
	F^{(2)}=
	\bma
	0&0&0&0\\
	0&0&2h_{0,1}h_{1,0}^*&0\\
	0&2h_{0,1}^*h_{1,0}&0&2h_{1,0}h_{1,1}^*\\
	0&0&2h_{1,0}^*h_{1,1}&0\\
	\ema,
\end{eqnarray}
where the entries satisfy $\sum_{j,k}\abs{g_{j,k}}^2=\sum_{j,k}\abs{h_{j,k}}^2=1$.

We need find the $g_{jk}$, $h_{jk}$ and $c_k$'s, such that
\begin{eqnarray}
	\label{eq:Xr-uczrucz}
	&&\r-(U_{CN}U_{CP}U_{CN})\r (U_{CN}U_{CP}U_{CN})^\dg=c_1F^{(1)}+c_2F^{(2)},
	\\
	\label{eq:Xc1+c2}
	&&c_1+c_2=\sqrt{2}\sin\frac{\t}{2}
\end{eqnarray}
holds for any $\r$, i.e.,
\begin{subnumcases}
	{\qquad\qquad\qquad\qquad\qquad\qquad}
	\label{eqset:Xcef1}
	a_{0,1}a_{1,0}^*\frac{(1-e^{-i\t })}{2}=c_1g_{0,1}g_{1,0}^*+c_2h_{0,1}h_{1,0}^*
	\\
	a_{0,0}a_{1,0}^*\frac{(1-e^{-i\t })}{2}=c_1g_{0,0}g_{1,0}^*
	\\
	a_{1,1}a_{1,0}^*\frac{(1-e^{-i\t })}{2}=c_2h_{1,1}h_{1,0}^*
	\\
	\label{eqset:Xnormalization}
	\sum_{j,k}\abs{g_{j,k}}^2=\sum_{j,k}\abs{h_{j,k}}^2=\sum_{j,k}\abs{a_{j,k}}^2=1
	\\
	\label{eqset:Xcef6}
	c_1+c_2=\sqrt{2}\sin\frac{\t}{2},\quad
	c_k\geq 0,\; \mbox{for}\; k=1,2		
\end{subnumcases}	
holds for any $a_{j,k}$. We consider two margin cases for the coefficients $a_{j,k}$, which will be used later.
\begin{enumerate}
	\item From (\ref{eq:r-XrX}), we have $\norm{\r-(U_{CN}U_{CP}U_{CN})\r (U_{CN}U_{CP}U_{CN})^\dg}_{W_1}=0$ when $a_{1,0}=0$ or $a_{1,0}=1$. Using  (\ref{eq:E(I,X)def}), it can be obtained that $c_1+c_2\leq  \sqrt{2}\sin\frac{\t}{2}$. 
	\item If $a_{0,0}=a_{1,1}=0$, we can choose $c_1=c_2=\frac{1}{\sqrt{2}}\sin\frac{\t}{2}$, $g_{0,0}=h_{1,1}=0$, $g_{1,0}^*=h_{0,1}=\frac{1}{\sqrt{2}}a_{0,1}$, $g_{0,1}^*=h_{1,0}=a_{1,0}\exp\{i(\frac{\t}{2}-\frac{\pi}{2})\}$ and $g_{1,1}=h_{0,0}=\sqrt{1-\frac{1}{2}\abs{a_{0,1}}^2-\abs{a_{1,0}}^2}$, such that (\ref{eqset:Xcef1})-(\ref{eqset:Xcef6}) hold. 
\end{enumerate}
	Next we only consider the case for
	\begin{eqnarray}
		\label{neq:a10}
	&&a_{1,0}\neq 0,1,\\
	\label{neq:a11}
	&&a_{0,0}\neq 0 \; \mbox{or}\; a_{1,1}\neq 0.	
	\end{eqnarray}

Using the normalization condition in (\ref{eqset:Xnormalization}), we can parameterize the coefficients $g_{j,k}$ and $h_{j,k}$ by $\a_j,\b_j\in[0,\frac{\pi}{2}]$, i.e., 
\begin{eqnarray}
	\label{eq:Xejk1}
	&&g_{0,0}=\cos\a_0\sin\a_1e^{i\a_2},
	\quad
	g_{0,1}=\cos\a_0\cos\a_1,
	\\
	\label{eq:Xejk2}
	&&g_{1,0}=\sin\a_0\sin\a_3e^{i\a_6},
	\quad
	g_{1,1}=\sin\a_0\cos\a_3e^{i\a_4},
\end{eqnarray}
and
\begin{eqnarray}
	\label{eq:Xfjk1}
	&&h_{0,0}=\sin\b_0\cos\b_3e^{i\b_4},
	\quad
	h_{0,1}=\cos\b_0\cos\b_1,
	\\
	\label{eq:Xfjk2}
	&&h_{1,0}=\sin\b_0\sin\b_3e^{i\b_6},
	\quad
	h_{1,1}=\cos\b_0\sin\b_1e^{i\b_2},
\end{eqnarray}
such that
\begin{subnumcases}{\qquad}
	\label{eqset:Xcab1}
	a_{0,1}a_{1,0}^*\frac{(1-e^{-i\t })}{2}=c_1\cos\a_0\cos\a_1\sin\a_0\sin\a_3e^{-i\a_6}
	+
	c_2\cos\b_0\cos\b_1
	\sin\b_0\sin\b_3e^{-i\b_6},
	\\
	\label{eqset:Xcab2}
	a_{0,0}a_{1,0}^*\frac{(1-e^{-i\t })}{2}=c_1\cos\a_0\sin\a_1e^{i\a_2}
	\sin\a_0\sin\a_3e^{-i\a_6},
	\\
	\label{eqset:Xcab3}
	a_{1,1}a_{1,0}^*\frac{(1-e^{-i\t })}{2}=c_2\cos\b_0\sin\b_1e^{i\b_2}
	\sin\b_0\sin\b_3e^{-i\b_6},
	\\
	\label{eqset:Xcab4}
	c_1+c_2=\sqrt{2}\sin\frac{\t}{2}.
\end{subnumcases}
For any $a_{j,k}\in\bbC$, we can always choose appropriate phase $\a_2,\b_2,\a_6,\b_6$ so that the phase of $a_{j,k}$ can be satisfied.  Without loss of generality, we can only consider the case that $a_{j,k}$ are nonnegative real value, i.e.,
\begin{eqnarray}
	\label{eq:Xconditionbjk}
	a_{j,k}\geq0 \quad \mbox{and} \quad \sum_{j,k}a_{j,k}^2=1.	
\end{eqnarray}
Using (\ref{neq:a10}) and (\ref{neq:a11}),  we assume that $a_{1,0}\in (0,1)$, and $a_{0,0}\in(0,1)$ or $a_{1,1}\in(0,1)$ here.

Now we show that (\ref{eqset:Xcab1})-(\ref{eqset:Xcab4}) is viable by choosing appropriate parameters. We set $\a_2=0$ and $\a_j=\b_j$, for $j=0,1,2,3,6$. 
From (\ref{eqset:Xcab2})-(\ref{eqset:Xcab4}), we assume
\begin{eqnarray}
	\label{eq:Xb01/c1=b10/c2}
	c_1=\frac{\sqrt{2}\sin\frac{\t}{2}a_{0,0}}{a_{0,0}+a_{1,1}},
	c_2=\frac{\sqrt{2}\sin\frac{\t}{2}a_{1,1}}{a_{0,0}+a_{1,1}}.
\end{eqnarray}
Using (\ref{eqset:Xcab2}), (\ref{eqset:Xcab3}) and (\ref{eq:Xb01/c1=b10/c2}), Eq.  (\ref{eqset:Xcab1}) can be equivalently  transformed into
\begin{eqnarray}
	\label{eq:Xb01cota}
	a_{0,0}\cot\a_1+a_{1,1}\cot\b_1=a_{0,1}.
\end{eqnarray}
It holds by choosing appropriate parameters  $\a_1=\b_1$. 
Next we show that equations (\ref{eqset:Xcab2}) and (\ref{eqset:Xcab3}) can also be satisfied.  
First we use (\ref{eq:Xb01cota}) to obtain that 
\begin{eqnarray}
	\label{eq:Xsina1}
	\sin\a_1=\sin\b_1=\frac{a_{0,0}+a_{1,1}}{\sqrt{a_{0,1}^2+(a_{0,0}+a_{1,1})^2}}.
\end{eqnarray}
Then based on (\ref{eq:Xb01/c1=b10/c2}) and  (\ref{eq:Xsina1}), we perform the transformation on (\ref{eqset:Xcab2}) and (\ref{eqset:Xcab3})  to put the free parameters $\a_0=\b_0$ and $\a_3=\b_3$ on the lhs alone. They become the same equation as follows.
\begin{eqnarray}
	\label{eq:Xcosasina}
	\cos\a_0\sin\a_0\sin\a_3e^{-i\a_6}=e^{i(\frac{\pi}{2}-\frac{\t}{2})}\frac{a_{1,0}}{\sqrt{2}}\sqrt{a_{0,1}^2+(a_{0,0}+a_{1,1})^2},
\end{eqnarray}
Note that $\frac{a_{1,0}}{\sqrt{2}}\sqrt{a_{0,1}^2+(a_{0,0}+a_{1,1})^2}\in(0,\frac{1}{2}]$. So Eq.  (\ref{eq:Xcosasina}) can always be satisfied for any $a_{j,k}$ in (\ref{eq:Xconditionbjk}) by choosing $\a_6=\b_6=\frac{\t}{2}-\frac{\pi}{2}$ and  appropriate parameters $\a_0,\a_3,\b_0,\b_3$. Hence, (\ref{eqset:Xcab2}) and (\ref{eqset:Xcab3}) can be satisfied.

Based on (\ref{eq:XD1D2}) and (\ref{eq:Xejk1})-(\ref{eq:Xfjk2}), the above analysis prove the existence of the $c^{(k)}$ and $F^{(k)}$ in (\ref{eq:r-XrX}) and (\ref{eq:Xc1+c2}) for any $\r$.
It means that there is a kind of decomposition following the rule in (\ref{eq:E(I,X)def}), such that $c_1+c_2=\sqrt{2}$. Recall that the quantum Wasserstein distance between operations  $\cD(I,U_{CN}U_{CP}U_{CN})$ is defined by taking the minimization of $c_1+c_2$ over all decompositions in (\ref{eq:E(I,X)def}).  One can obtain that
\begin{eqnarray}
	\cD(I,U_{CN}U_{CP}U_{CN})\leq \sqrt{2}\sin\frac{\t}{2}.
\end{eqnarray}
Combining with (\ref{ieq: rangeX}), it holds that
\begin{eqnarray}
	\cD(I,U_{CN}U_{CP}U_{CN})= \sqrt{2}\sin\frac{\t}{2},
\end{eqnarray}	
which is the desired result.
\end{proof}

\bibliographystyle{unsrt}
\bibliography{W1_norm}
\end{document}